\documentclass[journal,twocolumn,web]{ieeecolor}

\usepackage{generic}
\usepackage{amsmath,amssymb,amsfonts}
\usepackage{graphicx}
\usepackage{hyperref}
\hypersetup{hidelinks=true}
\usepackage{textcomp}

\usepackage{float}
\usepackage{moreverb}
\usepackage{mathrsfs}
\usepackage{subeqnarray}
\usepackage{cases}
\usepackage{enumerate}
\usepackage{subfigure}
\usepackage{epsfig}
\usepackage{setspace}
\usepackage[noadjust]{cite}
\usepackage{tikz}
\usetikzlibrary{shapes.geometric, arrows}
\usepackage{graphicx}
\usepackage{enumerate}
\usepackage{subfigure}
\usepackage{epsfig}
\usepackage{setspace}
\usepackage{multirow}
\usepackage{algorithm}
\usepackage{algorithmic}
\usepackage{autobreak}

\usepackage{makecell}

\makeatletter

\newcommand{\Rmnum}[1]{\expandafter\@slowromancap\romannumeral #1@}
\makeatother

\usepackage{amsthm}
\newtheorem{problem}{Problem}
\newtheorem{definition}{Definition}
\newtheorem{theoremx}{Theorem}

\newtheorem{lemmax}{Lemma}
\newtheorem{remark}{Remark}
\newtheorem{assumption}{Assumption}
\newtheorem{corollary}{Corollary}
\newtheorem*{notation}{Notation}
\allowdisplaybreaks

\def\BibTeX{{\rm B\kern-.05em{\sc i\kern-.025em b}\kern-.08em
		T\kern-.1667em\lower.7ex\hbox{E}\kern-.125emX}}
\begin{document}
	\title{Fast Distributed Algorithm for Aggregative Games in Malicious Environment}
	\author{Kai-Yuan Guo, Yan-Wu Wang, \IEEEmembership{Senior Member, IEEE}, Xiao-Kang Liu, \IEEEmembership{Member, IEEE}, Zhi-Wei Liu, \IEEEmembership{Senior Member, IEEE}
		\thanks{This work is supported by the National Natural Science Foundation of China under Grants 62233006 and 62573202, the Key Research and Development Project of Hubei Province under Grant 2025BAB013, and the Fundamental Research Funds for the Central Universities under grant 2024BRA014.}
		\thanks{Kai-Yuan Guo, Yan-Wu Wang, Xiao-Kang Liu, Zhi-Wei Liu are with the School of Artificial Intelligence and Automation, Huazhong University of Science and Technology, Wuhan 430074, China, and also with the Key Laboratory of Image Processing and Intelligent Control, Huazhong University of Science and Technology, Ministry of Education, Wuhan 430074, China (e-mail: wangyw@hust.edu.cn).}}
	
	\maketitle
	
	\begin{abstract}
		This paper addresses the distributed Nash Equilibrium seeking problem for aggregative games, where legitimate players' decisions are affected by potential malicious players. To describe players' behavior, we introduce a novel heterogeneous trustworthiness probabilistic framework by employing stochastic trust observations. To mitigate the waste of communication and gradient computation, we utilize a compressible unbalanced network information matrix and a multi-round communication mechanism to develop a fast Nash equilibrium seeking algorithm for aggregative games with unbalanced directed networks. By integrating the multi-round communication mechanism and a trustworthiness broadcast mechanism, we embed our fast convergence algorithm into the heterogeneous trustworthiness probabilistic framework, yielding a resilient fast Nash equilibrium seeking algorithm. Theoretical analysis confirms the convergence of the algorithm. Comparative simulations verify the accuracy of our fast convergence algorithm, and validation simulations verify the resilience of the algorithm.
	\end{abstract}
	
	
	\section{INTRODUCTION}
	Aggregative games have recently received increasing attention due to their extensive applications, such as smart grid \cite{saad2012game}, market economy \cite{okuguchi2012theory}, and cyber-physical systems \cite{huang2021cyber}. In these applications, players participate in a competition to reduce their costs influenced by their own strategies as well as the sum of all players’ strategies (aggregate strategy). However, an individual player may only communicate with neighbors through a communication network and be oblivious to the aggregate strategy. \textcolor{blue}{This necessitates the determination of the Nash equilibrium (NE) using the information gathered from neighbors to describe a stable state where no legitimate player can reduce its cost by unilaterally deviating from its strategy.}
	
	Many studies are dedicated to solving the NE seeking problem in aggregative games with balanced and unbalanced networks. To address balanced networks, \cite{huang2022linearly} employs a multi-round communication mechanism to design a distributed algorithm for balanced undirected networks, ensuring linear convergence to the NE. \cite{bianchi2022fast} proposes a proximal-based algorithm that achieves linear convergence without requiring multiple rounds of communication. \cite{zhu2023distributed} applies the multi-round communication mechanism to achieve a linear convergence rate in time-varying and balanced directed networks. \cite{carnevale2024tracking} proposes a tracking-based algorithm for constrained aggregative games with balanced directed networks, also achieving a linear convergence rate. However, many communication networks in practice are unbalanced due to the heterogeneous of communication capabilities of agents \cite{nedic2014distributed}. To address unbalanced networks, \cite{02Fang2022Directed} proposes NE seeking algorithms for aggregative games by utilizing row- and column-stochastic weight matrices. \cite{zhu2022asynchronous} presents a fully distributed algorithm for asynchronous aggregative games with unbalanced networks based on column-stochastic weight matrices. However, neither of \cite{02Fang2022Directed} and \cite{zhu2022asynchronous} can achieve linear convergence.

    \textcolor{blue}{Apart from the network topology, malicious behaviors including attacks and deceptions may also complicate the NE seeking. In case of attacks, \cite{feng2020attack} designs a distributed algorithm for NE seeking in directed networks under Denial-of-Service attacks in continuous time, while \cite{cai2022resilient} introduces an observer-based strategy using extended state observers to filter out false-data-injection attacks on sensor and actuator signals. In case of deceptions, \cite{li2019efficient} shows that an informed player can leverage its “nested information” advantage to influence the behavior of an uninformed player. \cite{tang2024deception} considers a “deceiver” with privileged information by injecting other players’ exploration signals into its own action policy, lead to a “deceptive Nash equilibrium” that favors the deceiver. \cite{gadjov2023algorithm} proposes a resilient NE-seeking algorithm that mitigates deceptions by estimating all players’ strategies and using an observation graph alongside a communication graph. However, \cite{li2019efficient} and \cite{gadjov2023algorithm} rely on assumptions about the game structure, which limits their applicability to aggregative games. How to identify malicious players becomes a key for the NE seeking problem \cite{zhang2020game}.}

	Recently, a novel framework using physical information to characterize trust and classify malicious players has been introduced to consensus \cite{yemini2021characterizing} and optimization \cite{yemini2025resilient}\cite{dayi2024fast} to enhance resilience. In this framework, inter-agent trust observations is derived from the system's physical information to extract neighbors' trustworthiness. The observation is expressed as a trust value indicating the likelihood with which an agent can trust data received from neighbors. \textcolor{blue}{This framework does not rely on repetitive game structures or prior knowledge of malicious environments, and is therefore suitable for seeking the NE of aggregative games in malicious environment.} However, this framework cannot be directly integrated with existing algorithms for aggregative games. Firstly, \cite{yemini2025resilient} presents results based on the homogeneity assumption, with the findings being contingent upon the number of players. Secondly, classifying malicious players requires a certain number of communication rounds, but existing NE seeking algorithms with unbalanced network \cite{02Fang2022Directed}\cite{zhu2022asynchronous} include only one communication round and one gradient computation per iteration, which may result in inefficient use of gradient computations. Thirdly, unlike existing optimization \cite{yemini2025resilient}\cite{dayi2024fast} and consensus \cite{yemini2021characterizing} protocols, aggregative games have requirements for the initial values of players \cite{ye2021differentially}\cite{wang2024differentially}. However, excluding malicious players may result in initial values of legitimate players that do not meet the convergence conditions. In summary, directly combining the trustworthiness probabilistic framework with existing algorithms would result in large algorithmic errors, computation waste and even the inability to achieve NE.


	Inspired by these groundbreaking studies and the gap between the trustworthiness probabilistic framework and aggregative games, this paper aims to develop a resilient fast NE seeking algorithm for aggregative games with unbalanced network. The trustworthiness probabilistic framework is extended to heterogeneity and the maximum classification time is obtained, which is only related to the pair of neighbors with the weakest classification ability. Then, by using a constant/diminishing step-size and a compressible unbalanced network information matrix, a fast NE seeking algorithm with linear/sub-linear convergence rate is proposed, and the convergence error at each iteration is obtained. Finally, through the trustworthiness broadcast mechanism that solves the problem of initial value mismatch, together with the multi-round communication mechanism that reduces the number of gradient computations, a resilient fast NE seeking algorithm is proposed by combining the heterogeneous trustworthiness probabilistic framework and the fast NE seeking algorithm. Theoretical convergence result of the resilient fast NE seeking algorithm is also obtained. Comparative simulations show the accuracy and the convergence rate of the proposed fast convergence algorithm, and confirm the convergence of the resilient fast algorithm. The main contributions of this paper are summarized as follows.
	\begin{enumerate}[(1)]
			\item The probability of the maximum classification time within the trustworthiness probabilistic framework is established under heterogeneous scenarios. Different from \cite{yemini2021characterizing}\cite{yemini2025resilient}, which indicate that the maximum classification time depends on both the number of players and the uniform classification ability, it is revealed for heterogeneous players that the maximum correct classification time is primarily influenced by the pair of neighbors with the weakest classification ability and is independent of the number of players.
			\item A new fast NE seeking algorithm is proposed for aggregative games with unbalanced network. Different from \cite{huang2022linearly, bianchi2022fast, zhu2023distributed, carnevale2024tracking} that focus on balanced networks, and \cite{02Fang2022Directed}\cite{zhu2022asynchronous} that do not specify a convergence rate, the proposed fast convergence algorithm converges to the neighborhood of NE at a rate of 
			$O(a^k)$ with a constant step-size, and converges to the NE rigorously at a rate of $O(\frac{1}{(k+1)^{2r-1}})$ with a diminishing step-size by employing a multi-round communication mechanism.
			\item A resilient fast NE seeking algorithm is proposed by integrating the fast NE seeking algorithm into the heterogeneous trustworthiness probabilistic framework. To address the initial value mismatch problem, a trustworthiness broadcast mechanism is designed. Additionally, the algorithm's convergence is theoretically analyzed using a compressible unbalanced network information matrix.
	\end{enumerate}
	
	\begin{notation}
		Denote the set of real numbers as $\mathbb{R}$ and positive integers as $\mathbb{N}_+$. Denote $1_N \in \mathbb{R}^N$ as a column vector with all elements being $1$, and $e_i \in \mathbb{R}^N$ as a column vector with $i$-th element being $1$ while the others being $0$. $I_N$ is the identity matrix of dimension N. $[x]_i$ is the $i$-th element of vector $x \in \mathbb{R}^N$, and $[A]_{ij}$ is the entry on $i$-th row and $j$-th column of matrix $A$. The superscript $T$ denotes the transpose of vectors and matrices. $A^{-1}$ is the inverse of matrix $A$. $diag(A)$ is a diagonal matrix whose diagonal entries are the same as those of $A$. $\| \cdot \|$ is the Euclidean norm while the 1-norm is $\| \cdot \|_{1}$. $P_{\Omega_i}[\cdot]$ is the Euclidean projection. The partial derivative of $f(x,y)$ with respect to $x$ is denoted as $\nabla_{x} f(x,y)$. $o(\alpha)$ is a series of monomial terms at least quadratic in $\alpha$. $\mathbb{P}\{A\}$ is the possibility of the event $A$ occurring.
	\end{notation}
	
	\section{PRELIMINARIES and PROBLEM FORMULATION}
	\subsection{Aggregative Games with malicious players}
	An aggregative game contains legitimate players $\mathcal{L}=\{1,2,..., N\}$. The strategy sets of legitimate players are $\Omega_i\subseteq\mathbb{R}^{p}$, where $p$ is the dimension of each player's strategy set. Denote $x_i$ as the strategy of player $i$, and $x = [x_1^T, x_2^T, ..., x_{N}^T]^T$ as a stacked vector of legitimate players' strategies, $x\in \Omega\triangleq \prod_{i=1}^{N}\Omega_i$. The cost function of player $i$ is written as $f_i(x_i, \sigma(x))$, where $\sigma(x)=\sum\limits_{i=1}^{N}\phi_i(x_i)$ is the aggregate of legitimate players' strategies and $\phi_i: \Omega_i 
	\to \mathbb{R}$ is only known to player $i$. 
	
	Legitimate player $i\in \mathcal{L}$ attempts to identify malicious players and to minimize its cost function
	\begin{align}\label{Problem}
		\mathop{\min}_{x_i\in\Omega_i} f_i(x_i,\sigma(x)),
	\end{align}
	
	Since each player independently minimizes its own cost, the interaction among players gives rise to a non-cooperative game. Its solution can be described as a NE defined as follows.

 
	\begin{definition}[\cite{04bacsar1998dynamic}]\label{DefinitionNashEquilibrium}
		A strategy profile $x^*=(x_i^*, x_{-i}^*)\in \Omega$ is a Nash equilibrium for a game if for $i\in\mathcal{L}, f_i(x_i^*, x_{-i}^*) \leq f_i(x_i, x_{-i}^*)$, where $x_i\in\Omega_{i}$ and $x_{-i}=[x_1^T,x_2^T,...,x_{i-1}^T,x_{i+1}^T,...,x_{N}^T]^T$.
	\end{definition}
	
    A Nash equilibrium of an aggregative game can then be defined as $x^*=(x_1^*,...,x_N^*)$ where for all $i\in\mathcal{L}, f_i\big(x_i^*, \sigma(x^*)\big) \leq f_i\big(x_i, \sum_{j\neq i}^{N}\phi_{j}(x_j^*)+\phi_i(x_i)\big)$. \textcolor{blue}{Players seek the NE as it naturally captures the self-interested nature of players in the game.} The following assumption is introduced to ensure the existence of NE.
	
	\begin{assumption}[\cite{koshal2016distributed}]\label{AssumptionGameSets}
		For every player $i\in\mathcal{L}$, its strategy set $\Omega_i$ is nonempty, convex, and compact. The functions $f_i(x_i,\sigma(x))$ and $\phi_i(x_i)$ are continuously differentiable about $(x_i,\sigma(x))$ and $x_i$, respectively. Besides, $f_i(x_i,\sigma(x))$ is convex in $x_i$ on $\Omega_i$.
	\end{assumption}
	
	Define $G_i(x_i, \sigma(x)) \triangleq \nabla_{x_i}f_i(x_i, \sigma(x))$, this notation allows us to conduct the pseudo-gradient mapping $G(x) \triangleq [G_1(x_1, \sigma(x))^T, ..., G_N(x_N, \sigma(x))^T]^T$.
	
	\begin{assumption}[\cite{koshal2016distributed}]\label{AssumptionGameGradient}
		The pseudo-gradient mapping $G: \Omega \to \mathbb{R}^{Np}$ is $\mu$-strongly monotone on $\Omega$, i.e., there exists a constant $\mu>0$ such that
		$(G(x)-G(x'))^T(x-x')\geq \mu \|x-x' \|^2, \quad\forall x, x'\in\Omega.$
	\end{assumption}
 
\textcolor{blue}{
 \begin{assumption}[\cite{02Fang2022Directed}\cite{koshal2016distributed}]\label{AssumptionLips}
		There exist positive constants $L_2$, and $L_3$, such that 
		$$\|G(x,\sigma) - G(x,\sigma') \| \leq L_2 \| \sigma - \sigma' \|, \quad\forall \sigma, \sigma' \in \mathbb{R},$$
		$$\|\phi_i(x_i) - \phi_i(x_i') \| \leq L_3 \| x_i - x_i' \|, \quad\forall x_i, x_i' \in \Omega_i, \forall i \in \mathcal{N}.$$
	\end{assumption}
}

	\textcolor{blue}{Assumptions \ref{AssumptionGameSets}-\ref{AssumptionLips} are widely used in studies of aggregative games \cite{huang2022linearly}\cite{carnevale2024tracking}. Assumption \ref{AssumptionGameGradient} claims the uniqueness of the NE $x^*$, which satisfies $x^*=P_{\Omega}[x^*-\alpha G(x^*)], \forall \alpha>0$ \cite{facchinei2003finite}.} In aggregative games, non-cooperation arises from individual rationality. \textcolor{blue}{Hence, once others adopt Nash equilibrium strategies, player $i$ must do the same to avoid a higher cost, which encourages communication among players.} Suppose the legitimate players can communicate iteratively through a fixed unbalanced directed communication network that has presence of malicious players denoted as $\mathcal{M}=\{N+1, N+2, ..., N'\}$. The network is modeled by a digraph $\mathcal{G}=\{ \mathcal{N},\mathcal{E} \}$, comprising a non-empty player set $\mathcal{N}=\mathcal{L}\cup\mathcal{M}$ and an edge set $\mathcal{E} \subseteq \{(i,j),i,j\in\mathcal{N}\}$.  $(j,i)\in\mathcal{E}$ if player $i$ can receive information directly from player $j$. The in-neighbor set of player $i$ is defined as $\mathcal{N}_{i}=\{j|(j,i)\in\mathcal{E}\} \bigcup \{i\}$. The weight matrix $A'=[a_{ij}] \in \mathbb{R}^{N'\times{N'}}$ is defined such that $a_{ij}>0$ if $(j,i)\in\mathcal{E}$, $a_{ij}=0$ otherwise. A directed graph is strongly connected if there is a path from every node to every other node. To ensure every legitimate player participates in the game, the following assumption is introduced.
	
	\begin{assumption}[\cite{02Fang2022Directed}]\label{AssumptionCommunicationGraph}
		The communication sub-digraph $\mathcal{G}_{\mathcal{L}}$ induced by legitimate players is fixed and strongly connected. 
	\end{assumption}

	The purpose of malicious players is similar to that in \cite{yemini2025resilient}, which is to prevent legitimate players from seeking the NE. \textcolor{blue}{Malicious players are assumed to transmit arbitrary values of appropriate dimensions to their neighbors, without adhering to any optimization objective. Consequently, this paper avoids modeling malicious players' strategies via any assumed principle or game-theoretic rationale, even though such behaviors could, in theory, be structured in some scenarios.} Besides, $\mathcal{L}$ and $\mathcal{M}$ are defined for analysis, and the legitimate agents do not know which agents in the system are legitimate or malicious at beginning.

	\subsection{Trustworthiness values}
	Since players cannot discern whether their neighbors are legitimate or malicious, aggregate decisions may be influenced by malicious players. This subsection adopts the trustworthiness probabilistic framework, as introduced in \cite{yemini2021characterizing}, to characterize the behaviors of both legitimate and malicious players. Within this framework, the extraction of trust observation is defined as follows.
	\begin{definition}\label{DefinitionTrustValue}
		For a legitimate player $i\in\mathcal{L}$ and its neighbors $j\in\mathcal{N}_i$, a random variable $\tau_{ij} \in [0,1]$, is defined as the probability that player $j$ is a legitimate neighbor. A random sample of the variable at time $k$ is available to player $i$, and is denoted as $\tau_{ij,k}$.
	\end{definition}
	
	\begin{assumption}\label{AssumptionTrustValue}
		The following statements are valid.
		\begin{enumerate}[1)]
			\item For $i\in\mathcal{L}, j\in\mathcal{N}_i$, there exist constant $E_{ij}\triangleq \mathbb{E}[\tau_{ij,k}] - \frac{1}{2},$ such that 
			$E_{ij} >0, j\in \mathcal{L},$ and
			$E_{ij} <0, j\in \mathcal{M}.$
			\item The observed information $\{\tau_{ij,k}\}_{k}$ is independent for all $k\geq1$ and for all $i \in \mathcal{L}, j\in \mathcal{N}_i$. Besides, for any $i\in\mathcal{L}$ and $j\in\mathcal{M}$, the observed information $\{\tau_{ij,k}\}_{k}$ is identically distributed.
		\end{enumerate}
	\end{assumption}
    \textcolor{blue}{Assumption \ref{AssumptionTrustValue} is made based on \cite[discussion]{yemini2025resilient}, which shows that $\tau_{ij}$ can be obtained in practice by extracting physical information such as wireless signal characteristics \cite{gil2017guaranteeing}}.
	Denote $E_{min} = \min\{E_{ij}, i\in\mathcal{L}, j\in\mathcal{N}_i\}$. Assume that $\tau_{ij,k}>\frac{1}{2}$ indicates a legitimate transmission and $\tau_{ij,k}<\frac{1}{2}$ indicates a malicious transmission in a stochastic sense. Although incorrect classification at iteration $k$ is possible, correct classification is certain in terms of expectation. The objective of this paper is presented in the following problem.

	\begin{problem}
		Propose a resilient fast NE seeking algorithm for aggregative games with unbalanced directed networks, where each legitimate player is able to classify its neighbors as legitimate or malicious during the NE seeking process.
	\end{problem}
	
	\section{Fast NE seeking algorithm}\label{GameSection}

	This section propose a fast NE seeking algorithm for the game without malicious players, inspired by \cite{huang2022linearly}, where the multi-round communication mechanism is proposed to accelerate the NE seeking process with balanced networks. The proposed algorithm adopts this communication mechanism with a constant or diminishing step-size, as detailed in Algorithm \ref{Algorithm 2}.

	\begin{algorithm}
		\caption{The fast NE seeking algorithm with unbalanced network}
		\label{Algorithm 2}
		\begin{algorithmic}[1]
			\STATE\textbf{Input}: $ x_{i,1} \in \Omega_i, v_{i,1} = e_{i},  \sigma_{i,1} = \phi_i(x_{i,1})$.
			\STATE\textbf{Output}: Players' strategies $x_i, \forall i \in\mathcal{L}$.
			\FOR {$k = 1, 2, ..., K$}
			\FOR {each player $i \in \mathcal{L}$}
			\STATE Communicate with all neighbors $\mathcal{N}_i$ as Multi-round Communication Mechanism
			$$\sigma_{i,k,S_k}, v_{i,k,S_k} \leftarrow MRC (\sigma_{i,k}, v_{i,k}).$$
			\STATE Updates variables as the Gradient-based Update Process 
			\begin{align}
				v_{i,k+1}, &x_{i,k+1}, \sigma_{i,k+1} \nonumber\\ 
				\leftarrow &Update(\sigma_{i,k}, \sigma_{i,k,S_k}, v_{i,k,S_k}, x_{i,k}, i'). \nonumber
			\end{align}
			\ENDFOR
			\ENDFOR
			\STATE
			\STATE \textbf{function} Multi-round Communication Mechanism $MRC(\sigma_{i,k}, v_{i,k})$:
			\STATE Set $\sigma_{i,k,0} = \sigma_{i,k}$, $v_{i,k,0} = v_{i,k}$.
			\FOR {$l = 0, 1, ..., S_k-1$}
			\STATE Send $\sigma_{i,k,l}, v_{i,k,l}$ to out-neighbors.
			\STATE Receive $\sigma_{j,k,l}, v_{j,k,l}$ from $j\in\mathcal{N}_{i}$.
			\begin{align}
				&\sigma_{i,k,l+1} = \sum\limits_{j\in\mathcal{N}_{i}} a_{ij,k} \sigma_{j,k,l}, \label{Multi-estimate}\\
				&v_{i,k,l+1} = \sum\limits_{j\in\mathcal{N}_{i}}a_{ij,k}v_{j,k,l}, \label{Multi-eigenvalue}
			\end{align}
			\ENDFOR 
			\STATE \textbf{return} $\sigma_{i,k,S_k}, v_{i,k,S_k}$
			\STATE
			\STATE \textbf{function} Gradient-based Update Process $Update(\sigma_{i,k}, \sigma_{i,k,S_k}, v_{i,k,S_k}, x_{i,k}, i')$:
			\STATE Player $i$ updates its variables
			\begin{align}
				& v_{i,k+1} = v_{i,k,S_k}, \label{gradientbasedupdate1}\\
				&x_{i,k+1} =P_{\Omega_i} [x_{i,k} - \alpha_k G_i(x_{i,k},\sigma_{i,k})], \label{gradientbasedupdate2}\\
				&\sigma_{i,k+1} = \sigma_{i,k,S_k} + \frac{\phi_i(x_{i,k+1})}{[v_{i,k+1}]_{i}} - \frac{\phi_i(x_{i,k})}{[v_{i,k}]_{i}}\label{gradientbasedupdate3},
			\end{align}
			\STATE \textbf{return} $v_{i,k+1}, x_{i,k+1}, \sigma_{i,k+1}$
		\end{algorithmic}
	\end{algorithm}

	In the multi-round communication mechanism, players engage in communication with their neighbors $S_k$ times, where $S_k$ is pre-set based on the step-size. During each communication, player $i$ sends information to its neighbors (including its estimate of the aggregate decisions denoted by $\sigma_{i,k,l}$, and its estimate of the left eigenvalue of the unbalanced network denoted by $v_{i,k,l}$), receives $\sigma_{j,k,l}, v_{j,k,l}$ from neighbors, and calculates their weighted average to obtain intermediate $\sigma_{i,k,S_k}, v_{i,k,S_k}$. Afterward, $i$ will obtain $v_{i,k+1}$, $\sigma_{i,k+1}$ and $x_{i,k+1}$ through the gradient-based update process. \textcolor{blue}{In Algorithm \ref{Algorithm 2}, \eqref{Multi-eigenvalue} and \eqref{gradientbasedupdate1} are used to estimate the weight of each player, a method widely applied in handling unbalanced directed network \cite{chen2024achieving}}, and \eqref{gradientbasedupdate3} is to estimate the aggregate decision formed by legitimate players. The weights $a_{ij,k}$ are constructed as follows.
    \begin{align}\label{weights_nomalicious}
		a_{ij,k} = \left\{
		\begin{aligned}
			&\frac{1}{|\mathcal{N}_{i,k}|}, j\in\mathcal{N}_{i}\\
			&0, j\notin \mathcal{N}_{i}
		\end{aligned}
		\right., 
	\end{align}
    a weight matrix $A\in \mathbb{R}^{N\times{N}}$ can then be obtained for the sub-digraph $\mathcal{G}_{\mathcal{L}}$ induced by legitimate players, \textcolor{blue}{with the properties presented in Lemma \ref{Lemma1} where $\theta$ is related to the number of players and network topology \cite{nedic2014distributed}.}
	
	\begin{lemmax}[\cite{xi2018linear}]\label{Lemma1}
		Suppose that Assumption \ref{AssumptionCommunicationGraph} holds and the weight matrix $A$ is row-stochastic. Denote $V_{k+1} = A^k$, $V_{\infty} = \lim\limits_{k \to \infty}A^{k}$, $\hat{V}_{k} = diag(V_{k})$ and $\hat{V}_{\infty} = diag(V_{\infty})$. There exist constants $\gamma>0$ and $\theta$ satisfying $0<\theta<1$, such that $$\|\hat{V}_k^{-1}-\hat{V}_{\infty}^{-1}\| \leq \rho^2 \gamma \theta^k,$$ where $\rho = \sup_{k}\|\hat{V}_{k}^{-1}\|$,.
	\end{lemmax}
	
	Denote $V_k=[v_{1,k},v_{2,k}...,v_{N,k}]^T$. Define
	$\sigma_k = [\sigma_{1,k}, \sigma_{2,k}, ..., \sigma_{N,k}]^T$ as the column vector of all players' estimate of the aggregate strategies at iteration $k$, $\sigma(x_k) = \sum\limits_{i=1}^N \phi_i(x_{i,k})$ as the true value of aggregate strategy, and $\sigma^*=\sum\limits_{i=1}^N \phi_i(x_i^*)$ as the aggregate NE strategy. Additionally, denote $\Phi_k = [\phi_{1,k}, \phi_{2,k}, ..., \phi_{N,k}]^T$. With these definitions, Algorithm \ref{Algorithm 2} can be written in the following compact forms
	\begin{align}
		&V_{k+1} = A^{S_k}V_{k} \label{AlgorithmUpdate2};\\
		&x_{k+1} = P_{\Omega}[x_k - \alpha_k G(x_k, \sigma_k)] \label{AlgorithmUpdate1};\\
		&\sigma_{k+1} = A^{S_k}\sigma_{k} + \hat{V}_{k+1}^{-1}\Phi_{k+1} - \hat{V}_{k}^{-1}\Phi_{k}. \label{AlgorithmUpdate3}
	\end{align}
	
	According to Assumption \ref{AssumptionTrustValue}-1) and the definition of weight matrix $A$, there is a vector $v\in\mathbb{R}^N$ such that $v^TA=v^T$, $1_N^Tv=1$, and $[v]_i>0$ by the Perron-Frobenius theorem. 
	
	Define $\hat{V}_{\infty}$-weighted Euclidean norm $\|x\|_{\hat{V}_{\infty}} = \sqrt{x^T\hat{V}_{\infty}x}$ and $\hat{V}_{\infty}$-induced norm of matrix $\|A\|_{\hat{V}_{\infty}} = \sup_{x\neq0}\frac{\|Ax\|_{\hat{V}_{\infty}}}{\|x\|_{\hat{V}_{\infty}}}$. By \cite{bianchi2020nash}, there is $\rho(A-V_{\infty})\leq\|A-V_{\infty}\|_{\hat{V}_{\infty}}<1$ and $\sqrt{[v]_{min}} \|x\| \leq \|x\|_{\hat{V}_{\infty}} \leq \sqrt{[v]_{max}} \|x\|$, where $[v]_{min}$ and $[v]_{max}$ are the minimum and maximum component of $v$, respectively. Besides, define $\rho_1 = \|A-V_{\infty}\|_{\hat{V}_{\infty}}$ and $\rho_2 = \|I_N - V_{\infty}\|$.
	Furthermore, under Assumption \ref{AssumptionGameSets}, there exists constants $C$ and $\tilde{C}$ such that $\|x_k\|\leq C$ and $\|\Phi_k\|\leq \tilde{C}$.
	\begin{theoremx}\label{SectionGameTheorem1}
		Suppose that \textcolor{blue}{Assumptions \ref{AssumptionGameSets}, \ref{AssumptionGameGradient}, \ref{AssumptionLips}, and \ref{AssumptionCommunicationGraph} hold.} Let $S_k = S \geq 1$ and  $\alpha_k = \alpha$, $\forall k\in\mathbb{N}_+$, such that $\rho(\Gamma_{\alpha}) = \lambda_{max}(\Gamma_{\alpha})<1$.
		Then the sequence $\{x_k\}_{k\in\mathbb{N}_+}$ generated by Algorithm \ref{Algorithm 2} converges to the neighborhood of NE $x^*$ at a linear rate, i.e., 
		\begin{align}\label{SectionGameTheorem1Result}
			&\|x_{k+1} - x^*\|^2 
   \leq \rho(\Gamma_{\alpha})^k\big[\|x_1-x^*\|^2 + \|\sigma_1 - V_{\infty}\sigma_1\|_{\hat{V}_{\infty}}^2\big] \nonumber\\
   &+D_1 \frac{\theta^2[\rho(\Gamma_{\alpha})^k - (\theta^2)^k]}{\rho(\Gamma_{\alpha}) - (\theta^2)} + D_2\frac{\theta[\rho(\Gamma_{\alpha})^k - (\theta)^k]}{\rho(\Gamma_{\alpha}) - (\theta)} \nonumber\\
			&+ D_3\frac{1-\rho(\Gamma_{\alpha})^k}{1-\rho(\Gamma_{\alpha})}, 
		\end{align}
		where 
		\begin{align}\small
			\Gamma_{\alpha} = \begin{bmatrix}
				1-2\alpha\mu & \frac{\alpha L_2}{\sqrt{[v]_{min}}}\\
				\frac{\alpha L_2}{\sqrt{[v]_{min}}} & (\rho_1^2)^{S} \\
			\end{bmatrix}, \nonumber
		\end{align}
		$D_1 = 4L_2^2\|V_{\infty}\|^2\tilde{C}^2\rho^4\gamma^2\alpha^2+8[v]_{max} \rho_2^2\tilde{C}^2\rho^4\gamma^2$, $D_2 = 4L_2C\|V_{\infty}\|\tilde{C}\rho^2\gamma\alpha+2\rho_1^{S_k}\tilde{M} \sqrt{[v]_{max}}2\rho_2\tilde{C}\rho^2\gamma$, and $D_3 = 2\rho_1^{S}\tilde{M}\sqrt{[v]_{max}}\rho_2\rho L_3M\alpha+2[v]_{max}\rho_2^2\rho^2 L_3^2M^2\alpha^2$.
	\end{theoremx}
	
	To prove Theorem \ref{SectionGameTheorem1}, we first present the following three lemmas.
	\begin{lemmax}\label{SectionGameLemma2}
		The following equation holds true
		\begin{align}
			V_{\infty}\sigma_k = V_{\infty}\hat{V}_{k}^{-1}\Phi_k, \forall k\in\mathbb{N}_+.
		\end{align}
	\end{lemmax}
	The proof of Lemma \ref{SectionGameLemma2}
 is similar to \cite{02Fang2022Directed} and thus is omitted.
	\begin{lemmax} \label{SectionGameLemma3}
		Under Assumptions \ref{AssumptionGameSets}, \ref{AssumptionGameGradient}, \ref{AssumptionLips}, and \ref{AssumptionCommunicationGraph}, 
		\begin{enumerate}[1)]
			\item there exists a constant $\tilde{M}$ such that $\|\sigma_{k+1}-V_{\infty}\sigma_{k+1}\|_{\hat{V}_{\infty}}\leq \tilde{M}, \forall k\in\mathbb{N}_+$.
			\item there exists a constant $M$ such that $\|G(x_k, \sigma_k)\|<M, \forall k\in\mathbb{N}_+$.
			\item $\|\sigma_{k+1}-V_{\infty}\sigma_{k+1}\|_{\hat{V}_{\infty}}$ satisfies following inequality
			\begin{align}
				&\|\sigma_{k+1}-V_{\infty}\sigma_{k+1}\|_{\hat{V}_{\infty}}^2 \nonumber\\
				\leq&(\rho_1^2)^{S_k}\|\sigma_{k}-V_{\infty}\sigma_{k}\|_{\hat{V}_{\infty}}^2 + 2\rho_1^{S_k}\tilde{M}\sqrt{[v]_{max}}\rho_2\rho L_3M\alpha_k \nonumber\\
				&+ 2\rho_1^{S_k}\tilde{M} \sqrt{[v]_{max}}2\rho_2\tilde{C}\rho^2\gamma\theta^k +2[v]_{max}\rho_2^2\rho^2 L_3^2M^2\alpha_k^2 \nonumber\\
				&+ 8[v]_{max}\rho_2^2\tilde{C}^2\rho^4\gamma^2(\theta^2)^k.
			\end{align}
			where $\rho$ is defined in Lemma \ref{Lemma1}.
		\end{enumerate}
	\end{lemmax}
	\begin{proof}

		
		\textbf{1)}: Based on \eqref{AlgorithmUpdate3} and Lemma \ref{SectionGameLemma2},
		\begin{align}
			&\|\sigma_{k+1}-V_{\infty}\sigma_{k+1}\|_{\hat{V}_{\infty}} \nonumber\\
			=&\|A^{S_k}\sigma_{k}+\hat{V}_{k+1}^{-1}\Phi_{k+1}-\hat{V}_{k}^{-1}\Phi_{k} - V_{\infty}\hat{V}_{k+1}^{-1}\Phi_{k+1}\|_{\hat{V}_{\infty}} \nonumber\\
			\leq&\|A^{S_k}\sigma_{k} - V_{\infty}\hat{V}_{k}^{-1}\Phi_{k}\|_{\hat{V}_{\infty}} \nonumber\\
			&+ \|V_{\infty}\hat{V}_{k}^{-1}\Phi_{k} + \hat{V}_{k+1}^{-1}\Phi_{k+1} - \hat{V}_{k}^{-1}\Phi_{k} - V_{\infty}\hat{V}_{k+1}^{-1}\Phi_{k+1}\|_{\hat{V}_{\infty}}. \nonumber
		\end{align}
		
		For the first term $\|A^{S_k}\sigma_{k} - V_{\infty}\hat{V}_{k}^{-1}\Phi_{k}\|_{\hat{V}_{\infty}}$, by noticing that $V_{\infty}V_{\infty}=V_{\infty}$ and Lemma \ref{SectionGameLemma2}, it satisfies
		\begin{align}\nonumber \label{SectionGameLemma2Step1Inequality1}
			&\|A^{S_k}\sigma_{k} - V_{\infty}\hat{V}_{k}^{-1}\Phi_{k}\|_{\hat{V}_{\infty}} \nonumber\\
			=&\|(A^{S_k}-V_{\infty})(\sigma_{k}-V_{\infty}\sigma_{k})\|_{\hat{V}_{\infty}} \nonumber\\
			\leq&\|(A-V_{\infty})^{S_k}\|_{\hat{V}_{\infty}}\|\sigma_{k}-V_{\infty}\sigma_{k}\|_{\hat{V}_{\infty}} \nonumber\\
			\leq&(\rho_1)^{S_k}\|\sigma_{k}-V_{\infty}\sigma_{k}\|_{\hat{V}_{\infty}}.
		\end{align}
		
		For the second term $\|V_{\infty}\hat{V}_{k}^{-1}\Phi_{k} + \hat{V}_{k+1}^{-1}\Phi_{k+1} - \hat{V}_{k}^{-1}\Phi_{k} - V_{\infty}\hat{V}_{k+1}^{-1}\Phi_{k+1}\|$, it has 
		\begin{align}\label{SectionGameLemma2Step1Inequality2}
			&\|V_{\infty}\hat{V}_{k}^{-1}\Phi_{k} + \hat{V}_{k+1}^{-1}\Phi_{k+1} - \hat{V}_{k}^{-1}\Phi_{k} - V_{\infty}\hat{V}_{k+1}^{-1}\Phi_{k+1}\| \nonumber\\
			=&\|(I_N - V_{\infty})(\hat{V}_{k+1}^{-1}\Phi_{k+1} - \hat{V}_{k}^{-1}\Phi_{k})\| \nonumber\\
			\leq&\rho_2 (\|\hat{V}_{k+1}^{-1}\Phi_{k+1} - \hat{V}_{k+1}^{-1}\Phi_{k}\| + \|\hat{V}_{k+1}^{-1}\Phi_{k} - \hat{V}_{k}^{-1}\Phi_{k}\|) \nonumber\\
			\leq&\rho_2(\|\hat{V}_{k+1}^{-1}\|L_3\|x_{k+1}-x_{k}\| + \|\hat{V}_{k+1}^{-1} - \hat{V}_{k}^{-1}\|\|\Phi_k\|) \nonumber\\
			\leq&\rho_2\rho L_32C + \rho_2\tilde{C}2\rho^2\gamma\theta^k.
		\end{align}
		
		Thus, combining \eqref{SectionGameLemma2Step1Inequality1} and \eqref{SectionGameLemma2Step1Inequality2} gives 
		\begin{align}
			&\|\sigma_{k+1}-V_{\infty}\sigma_{k+1}\|_{\hat{V}_{\infty}} \nonumber\\
			\leq&(\rho_1)^{S_k}\|\sigma_{k}-V_{\infty}\sigma_{k}\|_{\hat{V}_{\infty}} + \sqrt{[v]_{max}}\rho_2\rho L_32C  \nonumber\\ &+\sqrt{[v]_{max}}\rho_2\tilde{C}2\rho^2\gamma\theta^k.\nonumber
		\end{align}
		Since $S_k\geq 1$, $\rho_1^{S_k} \leq \rho_1$,
		\begin{align}
			&\|\sigma_{k+1}-V_{\infty}\sigma_{k+1}\|_{\hat{V}_{\infty}} \nonumber\\
			\leq& \rho_1^k\|\sigma_1-V_{\infty}\sigma_1\| + \sqrt{[v]_{max}}\rho_2\rho L_32C \sum_{l=0}^{k-1} \rho_1^l \nonumber\\
			&+ \sqrt{[v]_{max}}\rho_2\tilde{C}2\rho^2\gamma \sum_{l=1}^{k} \rho_1^{k-l}\theta^l.
		\end{align}
		It is easy to see that $\|\sigma_{k+1}-V_{\infty}\sigma_{k+1}\|_{\hat{V}_{\infty}}$ by a constant $\tilde{M}$.
		
		\textbf{2)}: Since $\sigma(x_k) = \sum_{i=1}^N\phi_i(x_i^k), x_i^k\in\Omega_i$ and $\Omega_i$ is a compact set, $\sigma(x_k)$ is bounded. Then, based on the continuity of $G$, $G(x_k, \sigma(x_k))$ is bounded and it holds that
		\begin{align}\label{SectionGameLemma2Step2Inequality1}
			&\|G(x_k, \sigma_k)\| \nonumber\\
			\leq& \|G(x_k, \sigma_k) - G(x_k, V_{\infty}\sigma_k)\| + \|G(x_k, V_{\infty}\sigma_k)\| \nonumber\\
			\leq& L_2\|\sigma_k-V_{\infty}\sigma_k\| + \|G(x_k, V_{\infty}\sigma_k)\|, 
		\end{align}
		so, from Lemma \ref{SectionGameLemma2}, $\|V_{\infty}\sigma_k\|\leq \|V_{\infty}\| \|\hat{V}_k^{-1}\| \|\Phi_k\|$ holds, and $\|V_{\infty}\sigma_k\|$ is bounded. Then, $G(x_k, V_{\infty}\sigma_k)$ is bounded. Combining Lemma \ref{SectionGameLemma3}-1) with \eqref{SectionGameLemma2Step2Inequality1} gives the conclusion 2).
		
		\textbf{3)}: Because $\|G(x_k, \sigma_k)\|<M$, from \eqref{SectionGameLemma2Step1Inequality2},
		\begin{align}
			&\|V_{\infty}\hat{V}_{k}^{-1}\Phi_{k} + \hat{V}_{k+1}^{-1}\Phi_{k+1} - \hat{V}_{k}^{-1}\Phi_{k} - V_{\infty}\hat{V}_{k+1}^{-1}\Phi_{k+1}\| \nonumber\\
			\leq & \rho_2(\|\hat{V}_{k+1}^{-1}\|L_3\alpha_kM + \|\hat{V}_{k+1}^{-1} - \hat{V}_{k}^{-1}\|\|\Phi_k\|). \nonumber
		\end{align}
		Thus
		\begin{align}
			&\|\sigma_{k+1}-V_{\infty}\sigma_{k+1}\|_{\hat{V}_{\infty}} \nonumber\\
			\leq&(\rho_1)^{S_k}\|\sigma_{k}-V_{\infty}\sigma_{k}\|_{\hat{V}_{\infty}} + \sqrt{[v]_{max}}\rho_2\rho L_3M\alpha_k \nonumber\\
			&+ \sqrt{[v]_{max}}2\rho_2\tilde{C}\rho^2\gamma\theta^k. \nonumber
		\end{align}
		Then,
		\begin{align}\label{SectionGameLemma2Step1Inequality3}
			&\|\sigma_{k+1}-V_{\infty}\sigma_{k+1}\|_{\hat{V}_{\infty}}^2 \nonumber\\
			\leq&(\rho_1^2)^{S_k}\|\sigma_{k}-V_{\infty}\sigma_{k}\|_{\hat{V}_{\infty}}^2 +2\rho_1^{S_k}\|\sigma_{k}-V_{\infty}\sigma_{k}\|_{\hat{V}_{\infty}}\nonumber\\
			&\cdot(\sqrt{[v]_{max}}\rho_2\rho L_3M\alpha_k + \sqrt{[v]_{max}}2\rho_2\tilde{C}\rho^2\gamma\theta^k) \nonumber\\
			&+2(\sqrt{[v]_{max}}\rho_2\rho L_3M\alpha_k)^2 + 2(\sqrt{[v]_{max}}2\rho_2\tilde{C}\rho^2\gamma\theta^k)^2. \nonumber\\
			\leq&(\rho_1^2)^{S_k}\|\sigma_{k}-V_{\infty}\sigma_{k}\|_{\hat{V}_{\infty}}^2 + 2\rho_1^{S_k}\tilde{M}\sqrt{[v]_{max}}\rho_2\rho L_3M\alpha_k \nonumber\\
			&+ 2\rho_1^{S_k}\tilde{M} \sqrt{[v]_{max}}2\rho_2\tilde{C}\rho^2\gamma\theta^k +2[v]_{max}\rho_2^2\rho^2 L_3^2M^2\alpha_k^2 \nonumber\\
			&+ 8[v]_{max}\rho_2^2\tilde{C}^2\rho^4\gamma^2(\theta^2)^k.
		\end{align}
	\end{proof}
	
	\begin{lemmax} \label{SectionGameLemma4}
		Under Assumptions \ref{AssumptionGameSets}, \ref{AssumptionGameGradient}, \ref{AssumptionLips}, and \ref{AssumptionCommunicationGraph}, the sequences $\{\|x_k-x^*\|^2\}_{k\in\mathbb{N}_+}$ generated by Algorithm \ref{Algorithm 2} satisfies:
		\begin{align}
			&\|x_{k+1} - x^*\|^2 \nonumber\\
			\leq&(1-2\alpha_k\mu)\|x_k - x^*\|^2 +4\alpha_k^2M^2 \nonumber\\
			&+2\alpha_kL_2\frac{1}{\sqrt{[v]_{min}}}\|x_k - x^*\|\|\sigma_k - V_{\infty}\sigma_k\|_{\hat{V}_{\infty}} \nonumber\\
			&+ 4\alpha_kL_2C\|V_{\infty}\|\tilde{C}\rho^2\gamma\theta^k.
		\end{align}
	\end{lemmax}
	\begin{proof}
        From the update rule \eqref{AlgorithmUpdate1},
		\begin{align} \label{SectionGameLemma3Step1Inequality1}
			&\|x_{k+1}-x^*\|^2\nonumber\\
			\leq& \|x_k-x^*\|^2 + \alpha_k^2\|G(x_k, \sigma_k) - G(x^*, \sigma^*)\|^2 \nonumber\\
			&-2\alpha_k(x_k-x^*)^T(G(x_k, \sigma_k) - G(x_k, 1_N\sigma(x_k))) \nonumber\\
			&-2\alpha_k(x_k-x^*)^T(G(x_k, 1_N\sigma(x^k)) - G(x^*, \sigma(x^*))) \nonumber\\
			\leq& (1-2\alpha_k \mu)\|x_k-x^*\|^2 + \alpha_k^2\|G(x_k, \sigma_k) - G(x^*, \sigma^*)\|^2 \nonumber\\
			&-2\alpha_k(x_k - x^*)^T(G(x_k,\sigma_k) - G(x_k, 1_N\sigma(x_k))).
		\end{align}
		The second inequality is derived by Assumption \ref{AssumptionGameGradient}.
		
		It is easy to see that $\|G(x_k, \sigma_k) - G(x^*, \sigma^*)\|^2 \leq 4M^2$. By noticing $V_{\infty}\hat{V}_{\infty}^{-1} = 1_N1_N^T$ and Lemma \ref{SectionGameLemma2}, 
		\begin{align}\label{SectionGameLemma3Step3Inequality1}
			&-2\alpha_k(x_k-x^*)^T(G(x_k, \sigma_k) - G(x_k, \sigma(x_k))) \nonumber\\
			\leq& 2\alpha_k\|x_k-x^*\|\|G(x_k, \sigma_k) - G(x_k, 1_N\sigma(x_k))\| \nonumber\\
			\leq& 2\alpha_k\|x_k-x^*\|L_2(\|\sigma_k - V_{\infty}\sigma_k\| + \|V_{\infty}\sigma_k - 1_N\sigma(x_k)\|) \nonumber\\
			\leq& 2\alpha_k\|x_k-x^*\|L_2(\|\sigma_k - V_{\infty}\sigma_k\| + \|V_{\infty}\|\tilde{C}\rho^2\gamma\theta^k) \nonumber\\
			\leq& 2\alpha_kL_2\|x_k-x^*\|\|\sigma_k - V_{\infty}\sigma_k\| + 4\alpha_kCL_2 \|V_{\infty}\|\tilde{C}\rho^2\gamma\theta^k.
		\end{align}
		
		Taking \eqref{SectionGameLemma3Step3Inequality1} into \eqref{SectionGameLemma3Step1Inequality1}, gives
		\begin{align}
			&\|x_{k+1} - x^*\|^2 \nonumber\\
			\leq& (1-2\alpha_k\mu)\|x_k - x^*\|^2 + 4\alpha_k^2M^2 \nonumber\\
			&+2\alpha_kL_2\|x_k - x^*\|\|\sigma_k - V_{\infty}\sigma_k\| \nonumber\\
			& + 4\alpha_kL_2C\|V_{\infty}\|\tilde{C}\rho^2\gamma\theta^k. \nonumber\\
			\leq&(1-2\alpha_k\mu)\|x_k - x^*\|^2 +4\alpha_k^2M^2 \nonumber\\
			&+2\alpha_kL_2\frac{1}{\sqrt{[v]_{min}}}\|x_k - x^*\|\|\sigma_k - V_{\infty}\sigma_k\|_{\hat{V}_{\infty}} \nonumber\\
			&+ 4\alpha_kL_2C\|V_{\infty}\|\tilde{C}\rho^2\gamma\theta^k.
		\end{align}
	\end{proof}
	
	Now it is ready to prove Theorem \ref{SectionGameTheorem1}.
	\begin{proof}
		Based on Lemmas \ref{SectionGameLemma2} and \ref{SectionGameLemma3}, it has
		\begin{align}
			&\|x_{k+1} - x^*\|^2 + \|\sigma_{k+1} - V_{\infty}\sigma_{k+1}\|_{\hat{V}_{\infty}}^2 \nonumber\\
			\leq &(1-2\alpha_k\mu)\|x_k - x^*\|^2 +4\alpha_k^2M^2 \nonumber\\
			&+ (\rho_1^2)^{S}\|\sigma_k - V_{\infty}\sigma_k\|_{\hat{V}_{\infty}}^2 \nonumber\\
			&+2\alpha_kL_2\frac{1}{\sqrt{[v]_{min}}}\|x_k-x^*\|\|\sigma_k - V_{\infty}\sigma_k\|_{\hat{V}_{\infty}} \nonumber\\
			&+ 4L_2C\|V_{\infty}\|\tilde{C}\rho^2\gamma\theta^k\alpha_k \nonumber\\
			&+2\rho_1^{S}\tilde{M}\sqrt{[v]_{max}}\rho_2\rho L_3M\alpha_k \nonumber\\
			&+ 2\rho_1^{S}\tilde{M} \sqrt{[v]_{max}}2\rho_2\tilde{C}\rho^2\gamma\theta^k\nonumber\\
			&+2[v]_{max}\rho_2^2\rho^2 L_3^2M^2\alpha_k^2 + 8[v]_{max}\rho_2^2\tilde{C}^2\rho^4\gamma^2(\theta^2)^k, \nonumber \\
            \leq & \begin{bmatrix}
				\|x_k-x^*\| & \|\sigma_k - V_{\infty}\sigma_k\|_{\hat{V}_{\infty}}
			\end{bmatrix}
			\Gamma
			\begin{bmatrix}
				\|x_k-x^*\| \\
				\|\sigma_k - V_{\infty}\sigma_k\|_{\hat{V}_{\infty}}
			\end{bmatrix} \nonumber\\
			&+ (\theta^2)^k(8[v]_{max} \rho_2^2\tilde{C}^2\rho^4\gamma^2) \nonumber\\
			&+\theta^k(4L_2C\|V_{\infty}\|\tilde{C}\rho^2\gamma\alpha_k+2\rho_1^{S}\tilde{M} \sqrt{[v]_{max}}2\rho_2\tilde{C}\rho^2\gamma) \nonumber\\
			&+2\rho_1^{S}\tilde{M}\sqrt{[v]_{max}}\rho_2\rho L_3M\alpha_k \nonumber\\
			&+(2[v]_{max}\rho_2^2\rho^2 L_3^2M^2 + 4M^2)\alpha_k^2, \nonumber
		\end{align}
		where \begin{align}\label{Gamma}
			\Gamma_{\alpha_k} = \begin{bmatrix}
				1-2\alpha_k\mu & \alpha_kL_2\frac{1}{\sqrt{[v]_{min}}}\\
				\alpha_kL_2\frac{1}{\sqrt{[v]_{min}}} & (\rho_1^2)^{S} \\
			\end{bmatrix}.
		\end{align}
		Since $\Gamma_{\alpha_k}$ is symmetric matrix, $\rho(\Gamma_{\alpha_k}) = \|\Gamma_{\alpha_k}\|_2 < 1$. Hence,
		\begin{align}\label{SectionGameTheorem1Inequality1}
			&\|x_{k+1} - x^*\|^2 + \|\sigma_{k+1} - V_{\infty}\sigma_{k+1}\|_{\hat{V}_{\infty}}^2 \nonumber\\
			\leq & \rho(\Gamma_{\alpha_k})\big[\|x_k-x^*\|^2 + \|\sigma_k - V_{\infty}\sigma_k\|_{\hat{V}_{\infty}}^2\big] \nonumber\\
			&+ (\theta^2)^k(8[v]_{max} \rho_2^2\tilde{C}^2\rho^4\gamma^2) \nonumber\\
			&+\theta^k(4L_2C\|V_{\infty}\|\tilde{C}\rho^2\gamma\alpha_k+2\rho_1^{S}\tilde{M} \sqrt{[v]_{max}}2\rho_2\tilde{C}\rho^2\gamma) \nonumber\\
			&+2\rho_1^{S}\tilde{M}\sqrt{[v]_{max}}\rho_2\rho L_3M\alpha_k \nonumber\\
			&+(2[v]_{max}\rho_2^2\rho^2 L_3^2M^2 + 4M^2)\alpha_k^2.
		\end{align}
		Denote $D_1 = 8[v]_{max} \rho_2^2\tilde{C}^2\rho^4\gamma^2$, $D_2 = 4L_2C\|V_{\infty}\|\tilde{C}\rho^2\gamma\alpha_k+2\rho_1^{S_k}\tilde{M} \sqrt{[v]_{max}}2\rho_2\tilde{C}\rho^2\gamma$, and $D_3 = 2\rho_1^{S}\tilde{M}\sqrt{[v]_{max}}\rho_2\rho L_3M\alpha_k+(2[v]_{max}\rho_2^2\rho^2 L_3^2M^2 + 4M^2)\alpha_k^2$. 
		
		Since $\alpha_k = \alpha$, 
		\begin{align}
			&\|x_{k+1} - x^*\|^2\nonumber\\
			\leq&\|x_{k+1} - x^*\|^2 + \|\sigma_{k+1} - V_{\infty}\sigma_{k+1}\|_{\hat{V}_{\infty}}^2 \nonumber\\
			\leq & \rho(\Gamma_{\alpha})^k\big[\|x_1-x^*\|^2 + \|\sigma_1 - V_{\infty}\sigma_1\|_{\hat{V}_{\infty}}^2\big] \nonumber\\
			&+D_1\sum_{l=1}^{k}\rho(\Gamma_\alpha)^{k-l}(\theta^2)^l + D_2\sum_{l=1}^{k}\rho(\Gamma_\alpha)^{k-l}\theta^l \nonumber\\
			&+ D_3\sum_{l=1}^{k}\rho(\Gamma_\alpha)^{k-l} \nonumber\\
			\leq& \rho(\Gamma_{\alpha})^k\big[\|x_1-x^*\|^2 + \|\sigma_1 - V_{\infty}\sigma_1\|_{\hat{V}_{\infty}}^2\big] \nonumber\\
			&+D_1\frac{\theta^2[\rho(\Gamma_\alpha)^k - (\theta^2)^k]}{\rho(\Gamma_\alpha) - (\theta^2)} + D_2\frac{\theta[\rho(\Gamma_\alpha)^k - (\theta)^k]}{\rho(\Gamma_\alpha) - (\theta)} \nonumber\\
			&+ D_3\frac{1-\rho(\Gamma_\alpha)^k}{1-\rho(\Gamma_\alpha)}. \nonumber
		\end{align}
	\end{proof}
	
	\textcolor{blue}{Theorem \ref{SectionGameTheorem1} shows that the outputs of the algorithm are within the neighborhood of the NE. The time complexity for a player per iteration is $O(S N^2 + p)$, derived from $S$ rounds of communication each requiring $O(N^2)$ operations to compute weighted averages of $N$-dimensional vectors from up to $N$ neighbors, plus $O(p)$ operations for the local gradient update.}

    \textcolor{blue}{From Theorem \ref{SectionGameTheorem1}, $S$ will affect both the convergence error and rate. Denote $1-\epsilon_2 = 1-2\alpha_k \mu$, $\epsilon_1 = \alpha_k L_2 \frac{1}{\sqrt{[v]_{min}}}$, and $A = (\rho_1^2)^S$ in \eqref{Gamma}, the spectral radius $\rho(\Gamma)$ can be derived as $\rho(\Gamma) = \frac{1}{2} \big[(1-\epsilon_2+A) + \sqrt{(1-\epsilon_2+A)^2 - 4A(1-\epsilon_2) +4\epsilon_1^2}\big] =\frac{1}{2} \big[(1-\epsilon_2+A) + \sqrt{(1-\epsilon_2-A)^2+4\epsilon_1^2}\big]  =\frac{1}{2} \big[(1-\epsilon_2+A) + (1-\epsilon_2-A)\sqrt{1+\frac{4\epsilon_1^2}{(1-\epsilon_2-A)^2}}\big]$. Since $\sqrt{1+x^2} \leq (1+\frac{1}{2}x^2)$, there is  $\sqrt{1+\frac{4\epsilon_1^2}{(1-\epsilon_2-A)^2}} \leq 1+\frac{2\epsilon_1^2}{(1-\epsilon_2-A)^2}$. Thus, $\rho(\Gamma)=\frac{1}{2}\big[(1-\epsilon_2+A) + (1-\epsilon_2-A)(1+\frac{2\epsilon_1^2}{(1-\epsilon_2-A)^2}) \leq 1-\epsilon_2+\frac{\epsilon_1^2}{1-\epsilon_1-A}$, which means 
    \begin{align}\label{DiscussionConvergence}
        \rho(\Gamma_{\alpha}) \leq 1 - 2\alpha \mu + \frac{\alpha^2 L_2^2 \frac{1}{[v]_{min}}}{1-\alpha \frac{L_2}{\sqrt{[v]_{min}}} - (\rho_1^2)^S},
    \end{align}
    \begin{enumerate}
        \item From \eqref{SectionGameTheorem1Result}, $\lim\limits_{k\to\infty}\|x_{k+1} - x^*\|^2 \leq D_3\frac{1}{1-\rho(\Gamma_\alpha)} = \frac{2\rho_1^{S}\tilde{M}\sqrt{[v]_{max}}\rho_2\rho L_3M\alpha+(2[v]_{max}\rho_2^2\rho^2 L_3^2M^2+4M^2)\alpha^2}{1-\rho(\Gamma_\alpha)} \leq \frac{\rho_1^{S}\tilde{M}\sqrt{[v]_{max}}\rho_2\rho L_3M}{\mu} $ with sufficiently small $\alpha$. This means that with a larger communication round per iteration $S$, the theoretical upper bound becomes smaller.
        \item Denote the total number of communication instances as $n$ (assuming $n = kS$, where $k$ is the total number of iterations.). From \eqref{SectionGameTheorem1Result}, the convergence rate with respect to $n$ can be expressed as $\rho(\Gamma_{\alpha})^{1/S}$, which is influenced by the game's  parameters, the network topology, and the step-size $\alpha$. From \eqref{DiscussionConvergence} the upper bound of $\rho(\Gamma_{\alpha})^{1/S}$ can be approximated as $[a + \frac{b}{c - d^S}]^{1/S}$, where $a = 1 - 2\alpha \mu, b = \alpha^2 L_2^2 \frac{1}{[v]_{min}}, c = 1-\alpha \frac{L_2}{\sqrt{[v]_{min}}}, d = (\rho_1^2)$. Therefore, the convergence rate can either increase monotonically with $S$, or first decreasing and then increasing with $S$, indicating a trade-off between convergence rate and $S$.
    \end{enumerate}
}

	\begin{remark}
		\begin{enumerate}
			\item \cite{huang2022linearly}\cite{zhu2023distributed}\cite{carnevale2024tracking} establish results on balanced networks, and \cite{huang2022linearly}\cite{zhu2023distributed} also require global information to determine the minimum communication rounds of per iteration. In contrast, Algorithm \ref{Algorithm 2} establishes on unbalanced directed network and achieves fully distributed by allowing for arbitrary communication rounds, i.e., $S \geq 1$.
			\item  \textcolor{blue}{\cite{02Fang2022Directed}\cite{zhu2022asynchronous} establish results on unbalanced networks with a diminishing step-size, but neither provides a clear convergence rate. In contrast, Algorithm \ref{Algorithm 2} uses a constant step-size and incorporates a compressible unbalanced network information matrix $\Gamma_{\alpha}$, thereby achieves linear convergence.}
		\end{enumerate}
	\end{remark}
	
	Algorithm \ref{Algorithm 2} converges only to the neighborhood of the NE when using constant step-size. However, by choosing appropriate step-size and communication rounds, it can also achieve rigorous convergence with sub-linear convergence rate as shown in the following corollary.
	
	\begin{corollary}\label{SectionGameCorollary1}
		Suppose that Assumptions \ref{AssumptionGameSets}, \ref{AssumptionGameGradient}, \ref{AssumptionLips}, and \ref{AssumptionCommunicationGraph} hold. Let step-size $\alpha_k = \frac{1}{(k+1)^r}, \forall k\in\mathbb{N}_+$, $\frac{1}{2}<r<1$. Let $S_k \geq \frac{\ln\frac{1}{2}}{2\ln\rho_1}$.
		Then, the sequence $\{x_k\}_{k\in\mathbb{N}_+}$ generated by Algorithm \ref{Algorithm 2} converges to the NE $x^*$ at a rate of $O(\frac{1}{(k+1)^{2r-1}})$.
	\end{corollary}
	\begin{proof}
		By using \eqref{SectionGameLemma2Step1Inequality3},
		\begin{align}
			&\|\sigma_{k+1}-V_{\infty}\sigma_{k+1}\|_{\hat{V}_{\infty}}^2 \nonumber\\
			\leq&2(\rho_1^2)^{S_k}\|\sigma_{k}-V_{\infty}\sigma_{k}\|_{\hat{V}_{\infty}}^2 +2\big[2(\sqrt{[v]_{max}}\rho_2\rho L_3M\alpha_k)^2 \nonumber\\
			&+ 2(\sqrt{[v]_{max}}2\rho_2\tilde{C}\rho^2\gamma\theta^k)^2\big]. \nonumber
		\end{align}
		Besides, by using Lemma \ref{SectionGameLemma4}, 
		\begin{align}
			&\|x_{k+1} - x^*\|^2 + \|\sigma_{k+1} - V_{\infty}\sigma_{k+1}\|_{\hat{V}_{\infty}}^2 \nonumber\\
			\leq &(1-2\alpha_k\mu)\|x_k - x^*\|^2 \nonumber\\
			&+ 2(\rho_1^2)^{S_k}\|\sigma_k - V_{\infty}\sigma_k\|_{\hat{V}_{\infty}}^2 \nonumber\\
			&+2\alpha_kL_2\frac{1}{\sqrt{[v]_{min}}}\|x_k-x^*\|\|\sigma_k - V_{\infty}\sigma_k\|_{\hat{V}_{\infty}} \nonumber\\
			& + 4L_2C\|V_{\infty}\|\tilde{C}\rho^2\gamma\theta^k\alpha_k  + 4M^2\alpha_k^2\nonumber\\
			&+4[v]_{max}\rho_2^2\rho^2 L_3^2M^2\alpha_k^2 + 16[v]_{max}\rho_2^2\tilde{C}^2\rho^4\gamma^2(\theta^2)^k \nonumber\\
			=& \begin{bmatrix}
				\|x_k-x^*\| & \|\sigma_k - V_{\infty}\sigma_k\|_{\hat{V}_{\infty}}
			\end{bmatrix}
			\Gamma'_{\alpha_k}
			\begin{bmatrix}
				\|x_k-x^*\| \\
				\|\sigma_k - V_{\infty}\sigma_k\|_{\hat{V}_{\infty}}
			\end{bmatrix} \nonumber\\
			& + 4L_2C\|V_{\infty}\|\tilde{C}\rho^2\gamma\theta^k\alpha_k + 4M^2\alpha_k^2 \nonumber\\
			&+4[v]_{max}\rho_2^2\rho^2 L_3^2M^2\alpha_k^2 + 16[v]_{max}\rho_2^2\tilde{C}^2\rho^4\gamma^2(\theta^2)^k \nonumber\\
			\leq & \rho(\Gamma'_{\alpha_k})\big[\|x_k-x^*\|^2 + \|\sigma_k - V_{\infty}\sigma_k\|_{\hat{V}_{\infty}}^2\big] \nonumber\\
			&+(\theta^2)^k(16[v]_{max}\rho_2^2\tilde{C}^2\rho^4\gamma^2) \nonumber\\
			&+\theta^k(4L_2C\|V_{\infty}\|\tilde{C}\rho^2\gamma\alpha_k) \nonumber\\
			&+\alpha_k^2(4[v]_{max}\rho_2^2\rho^2 L_3^2M^2 + 4M^2), \nonumber
		\end{align}
		where
		\begin{align}
			\Gamma'_{\alpha_k} = \begin{bmatrix}
				1-2\alpha_k\mu & \alpha_kL_2\frac{1}{\sqrt{[v]_{min}}}\\
				\alpha_kL_2\frac{1}{\sqrt{[v]_{min}}} & 2(\rho_1^2)^{S_k} 
			\end{bmatrix}. \nonumber
		\end{align}
		Denote $d_1 = 16[v]_{max}\rho_2^2\tilde{C}^2\rho^4\gamma^2$, $d_2 = 4L_2C\|V_{\infty}\|\tilde{C}\rho^2\gamma\alpha_1$, and $d_3 = 4[v]_{max}\rho_2^2\rho^2 L_3^2M^2 + 4M^2$.
		
		By explicit computation of $\rho(\Gamma'_{\alpha_k})$ and Taylor expansion at $ \alpha_k = 0$, $\rho(\Gamma'_{\alpha_k})\leq 1-2\mu \alpha_k +o(\alpha_k)$ holds.
		Based on \eqref{SectionGameTheorem1Inequality1}, it can be derived that
		\begin{align}\label{SectionGameCorollary1Inequallity1}
			&\|x_{k+1} - x^*\|^2 \nonumber\\
			\leq& \big[\prod_{l=1}^{k}(1-2\mu\alpha_k+o(\alpha_k))\big]\big[\|x_1-x^*\|^2 + \|\sigma_1 - V_{\infty}\sigma_1\|^2\big] \nonumber\\
			&+d_1\sum_{l=1}^{k}\big[\prod_{m=l+1}^{k}(1-2\mu\alpha_m+o(\alpha_m))\big](\theta^2)^l \nonumber\\
			&+ d_2\sum_{l=1}^{k}\big[\prod_{m=l+1}^{k}(1-2\mu\alpha_m+o(\alpha_m))\big]\theta^l \nonumber\\
			&+ d_3\sum_{l=1}^{k}\big[\prod_{m=l+1}^{k}(1-2\mu\alpha_m+o(\alpha_m))\big]\alpha_l^2.
		\end{align}
		
		Because of $\frac{1}{2} < r < 1$, there is a constant $K_0$ such that $2\mu\alpha_k(k+1)+o(\alpha_k)(k+1)\geq1$ when $k\geq K_0$ . Thus, $1-2\mu \alpha_k +o(\alpha_k)\leq 1-\frac{1}{k+1}$ when $k\geq K_0$. For simplicity assume $K_0=1$. Thus, \eqref{SectionGameCorollary1Inequallity1} can be written as
		\begin{align}
			&\|x_{k+1} - x^*\|^2 \nonumber\\
			\leq& \big[\prod_{l=1}^{k}(1-\frac{1}{k+1})\big]\big(\|x_1-x^*\|^2 + \|\sigma_1 - V_{\infty}\sigma_1\|^2\big) \nonumber\\
			&+d_1\sum_{l=1}^{k}\big[\prod_{m=l+1}^{k}(1-\frac{1}{m+1})\big](\theta^2)^l \nonumber\\
			&+ d_2\sum_{l=1}^{k}\big[\prod_{m=l+1}^{k}(1-\frac{1}{m+1})\big]\theta^l \nonumber\\
			&+ d_3\sum_{l=1}^{k}\big[\prod_{m=l+1}^{k}(1-\frac{1}{m+1})\big]\frac{1}{(k+1)^{2r}}. \nonumber
		\end{align}
		Since $\prod_{m=l+1}^{k}\big(1-\frac{1}{m+1}\big) = \frac{l+1}{k+1}$, the following equations hold.
		\begin{align}
			\textcircled{\scriptsize{1}}&\prod_{l=1}^{k}(1-\frac{1}{k+1}) = \frac{1}{k+1}, \nonumber\\
			\textcircled{\scriptsize{2}}&\sum_{l=1}^{k}\big[\prod_{m=l+1}^{k}(1-\frac{1}{m+1})\big](\theta^2)^l = \sum_{l=1}^{k}\frac{l+1}{k+1}(\theta^2)^l \nonumber\\
			=& \frac{1}{k+1}\frac{\theta^2}{1-\theta^2} + \frac{1}{k+1}\frac{\theta^2(1-(\theta^2)^k)}{(1-\theta^2)^2} - \frac{(\theta^2)^{k+1}}{1-\theta^2},\nonumber \\
			\textcircled{\scriptsize{3}}&\sum_{l=1}^{k}\big[\prod_{m=l+1}^{k}(1-\frac{1}{m+1})\big]\theta^l = \sum_{l=1}^{k}\frac{l+1}{k+1}\theta^l \nonumber\\
			=& \frac{1}{k+1}\frac{\theta}{1-\theta} + \frac{1}{k+1}\frac{\theta(1-\theta^k)}{(1-\theta)^2} - \frac{\theta^{k+1}}{1-\theta},\nonumber\\
			\textcircled{\scriptsize{4}}&\sum_{l=1}^{k}\big[\prod_{m=l+1}^{k}(1-\frac{1}{m+1})\big]\frac{1}{(l+1)^{2r}} = \sum_{l=1}^{k}\frac{l+1}{k+1}\frac{1}{(l+1)^{2r}} \nonumber\\
			\leq& \frac{1}{k+1}\sum_{l=1}^{k}\frac{1}{(l+1)^{2r-1}} \leq \frac{1}{k+1} \int_{0}^{k} \frac{1}{(x+1)^{2r-1}} dx \nonumber\\
			=& \frac{1}{k+1} \big(\frac{(k+1)^{2-2r}-1}{2-2r}\big) = \frac{1}{2-2r}\big[\frac{1}{(k+1)^{2r-1}} - \frac{1}{k+1}\big]. \nonumber
		\end{align}
		Thus, 
		\begin{align}\label{SectionCro}
			&\|x_{k+1} - x^*\|^2 \nonumber\\
			\leq& \big(\|x_1-x^*\|^2 + \|\sigma_1 - V_{\infty}\sigma_1\|^2\big)\frac{1}{k+1} \nonumber\\
			&+d_1\big[\frac{1}{k+1}\frac{\theta^2}{1-\theta^2} + \frac{1}{k+1}\frac{\theta^2(1-(\theta^2)^k)}{(1-\theta^2)^2} - \frac{(\theta^2)^{k+1}}{1-\theta^2}\big] \nonumber\\
			&+ d_2\big[\frac{1}{k+1}\frac{\theta}{1-\theta} + \frac{1}{k+1}\frac{\theta(1-\theta^k)}{(1-\theta)^2} - \frac{\theta^{k+1}}{1-\theta}\big] \nonumber\\
			&+ d_3\frac{1}{2-2r}\big[\frac{1}{(k+1)^{2r-1}} - \frac{1}{k+1}\big].
		\end{align}
		As $\frac{1}{2}<r<1$, $0<2r-1< 1$. Hence, $\prod_{l=1}^{k}(1-\frac{1}{k+1}) \to 0$, $\sum_{l=1}^{k}\big[\prod_{m=l+1}^{k}(1-\frac{1}{m+1})\big](\theta^2)^l \to 0$, $\sum_{l=1}^{k}\big[\prod_{m=l+1}^{k}(1-\frac{1}{m+1})\big]\theta^l \to 0$,  and $\sum_{l=1}^{k}\big[\prod_{m=l+1}^{k}(1-\frac{1}{m+1})\big]\frac{1}{(k+1)^{2r}} \to 0$. So $\|x_{k+1} - x^*\|^2 \to 0$ at a speed of $O(\frac{1}{(k+1)^{2r-1}})$.
	\end{proof}
	
	\begin{remark}
		Different from \cite{02Fang2022Directed} where convergence to the NE is analyzed via series convergence, Corollary \ref{SectionGameCorollary1} not only confirms rigorous convergence but also quantifies the convergence error at iteration $k$. This is achieved by constructing a compressible unbalanced network information matrix $\Gamma_{\alpha_k}$, and applying the Taylor expansion to the norm of this compressive matrix.
	\end{remark}
	
	The two explicit squared errors \eqref{SectionGameTheorem1Result} and \eqref{SectionCro} in Theorem \ref{SectionGameTheorem1} and Corollary \ref{SectionGameCorollary1} are pivotal for establishing the convergence of the resilient fast NE seeking algorithm for the game with malicious players in the following section.

	\section{Resilient NE seeking Algorithm}

    This section extend the fast NE seeking algorithm to resilient by incorporating the trustworthiness probabilistic framework. Specifically, this section introduces the classification and isolation procedure, the trustworthiness broadcast mechanism, and the resilient fast NE seeking algorithm, together with their theoretical analysis and convergence results.

	\subsection{The classification and isolation of players}\label{ClassificationSection}
	
	This subsection presents how correct classification is achieved within the heterogeneous trustworthiness probabilistic framework.
	
	In \cite{yemini2021characterizing}\cite{yemini2025resilient}, player $i$ performs one extraction of $\tau_{ij}$ and one gradient calculation in each iteration. To reduce the number of gradient calculations and inspired by the design of delay updating setting in \cite{yemini2025resilient} and the multi-round communication mechanism in \cite{jakovetic2014fast}, consider that a player performs $S_k$ extractions per iteration as trust values, and denote these extractions as $\tau_{ij,k,l}, l= 0, 1, ..., S_k-1$. The trustworthiness $\beta_{ij,k}$ is the sum over a history of the trust observations $\tau_{ij,k,l}$, i.e.
	\begin{align}
		\beta_{ij,k+1} = \sum_{m=1}^{k} \sum_{l=0}^{S_m-1} (\tau_{ij,m,l}-\frac{1}{2}),
	\end{align}
	$\beta_{ij,k}\geq0$ implies that $i$ classifies $j$ as a legitimate neighbor, conversely, $i$ classifies $j$ as a malicious neighbor. Define trust in-neighbors $\mathcal{N}_{i,k} =\{j\in\mathcal{N}_i:\beta_{ij,k}\geq0\}$. If $i$ classifies $j$ as a malicious player, $i$ will still accept information from $j$ (as this classification may be incorrect), but will not use $j$'s information by adjusting the weights based on the revised form of \eqref{weights_nomalicious} as following,
	\begin{align}\label{weights}
		a_{ij,k} = \left\{
		\begin{aligned}
			&\frac{1}{|\mathcal{N}_{i,k}|}, j\in\mathcal{N}_{i,k}\\
			&0, j\notin \mathcal{N}_{i,k}
		\end{aligned}
		\right.. 
	\end{align}
	Similar to \cite[Lemma 2]{yemini2021characterizing}, the following lemma can be obtained.
	\begin{lemmax}
		Under Assumption \ref{AssumptionTrustValue}, the following statements hold.
		\begin{enumerate}[1)]
			\item The probability that player $i\in\mathcal{L}$ classifies a legitimate neighbor $j\in \mathcal{N}_i\cap\mathcal{L}$ as a malicious one at iteration $k+1$ satisfies
			\begin{align}
				\mathbb{P}\{\beta_{ij,k+1}<0, i\in\mathcal{L}, j\in\mathcal{N}_i\cap\mathcal{L}\} \leq e^{-2E_{ij}^2\sum_{m=1}^{k}S_m}. \nonumber
			\end{align}
			\item The probability that player $i\in\mathcal{L}$ classifies a malicious neighbor $j\in \mathcal{N}_i\cap\mathcal{M}$ as a legitimate one at iteration $k+1$ satisfies
			\begin{align}
				\mathbb{P}\{\beta_{ij,k+1}\geq0, i\in\mathcal{L}, j\in\mathcal{N}_i\cap\mathcal{M}\} \leq e^{-2E_{ij}^2\sum_{m=1}^{k}S_m}. \nonumber
			\end{align}
		\end{enumerate}
	\end{lemmax}

	Define $T_{ij}$ as follows:
	for $i\in\mathcal{L}, j\in\mathcal{N}_i\cap\mathcal{L}$, $T_{ij}$ satisfies $\beta_{ij,T_{ij}-1}<0$ and $\beta_{ij,k}\geq0, \forall k\geq T_{ij}$; and
	for $i\in\mathcal{L}, j\in\mathcal{N}_i\cap\mathcal{M}$, $T_{ij}$ satisfies $\beta_{ij,T_{ij}-1}\geq0$ and $\beta_{ij,k}<0, \forall k\geq T_{ij}$. Denote the order statistic of correct classification time as $\{T^{(1)}, T^{(2)}, ..., T^{(\sum_{i=1}^{N'} |\mathcal{N}_i|)}\}$, denote $T_{f} = T^{(\sum_{i=1}^{N'} |\mathcal{N}_i|)}$,
	the following theorem verifies the existence of $T_{f}$ and investigates the probability of the correct classification time.
	\begin{theoremx}\label{SectionTrustTheorem1}
		Under Assumption \ref{AssumptionTrustValue}, the random finite time $T_{f}$ always exists, regardless of the design of $\{S_k\}_{k=1}^{K}$. Furthermore, the probability of $T_f$ satisfies
		\begin{align}
			\mathbb{P}\{T_{f}=k\}&\leq e^{-2E_{min}^2\sum_{m=1}^{k-1}S_m}.\nonumber
		\end{align}
	\end{theoremx}
	\begin{proof}
		\textbf{1)}: The first conclusion can be derived by applying the approaches in \cite{yemini2025resilient} from a player-player perspective. For $i\in\mathcal{L}, j\in\mathcal{N}_i\cap\mathcal{L}$, there is $\sum_{k=1}^{\infty} \mathbb{P}\{\beta_{ij,k+1}<0\} \leq \sum_{k=1}^{\infty}e^{-2E_{ij}^2\sum_{m=1}^{k}S_m} \leq \sum_{k=1}^{\infty}e^{-2E_{ij}^2k} < \infty$. Based on Borel–Cantelli lemma, the event $\{\beta_{ij,k+1}<0, i\in\mathcal{L}, j\in\mathcal{N}_i\cap\mathcal{L}\}$ occurs only finitely often almost surely. Similarly, $\{\beta_{ij,k+1}>0, i\in\mathcal{L}, j\in\mathcal{N}_i\cap\mathcal{M}\}$ occurs only finitely often almost surely, too. Thus, a random finite time $T_{ij}$ exists regardless of the design of $S_k$. Then, the random finite time $T_{f} = T^{(\sum_{i=1}^{N'} |\mathcal{N}_i|)}$ always exists.
		
		\textbf{2)}: As to the probability of $T_{f}$, suppose $i\in\mathcal{L}, j\in\mathcal{N}_i\cap\mathcal{L}$. If $T_{ij} = k$, then iteration $k$ is the first time at which edge $(i,j)$ is permanently classified correctly, whereas at iteration $k-1$ it is still incorrect. For one legitimate–legitimate edge, “correct classification” means $\beta_{ij,k} \geq 0$, while “not yet correct” at $k-1$ means $\beta_{ij,k-1} < 0$. This yields $\{T_{ij}=k, i\in\mathcal{L}, j\in\mathcal{N}_i\cap\mathcal{L}\} \subseteq \{\beta_{ij,k-1}<0, i\in\mathcal{L}, j\in\mathcal{N}_i\cap\mathcal{L}\}$,  which means $\mathbb{P}\{T_{ij}=k, i\in\mathcal{L}, j\in\mathcal{N}_i\cap\mathcal{L}\} \leq \mathbb{P}\{\beta_{ij,k-1}<0, i\in\mathcal{L}, j\in\mathcal{N}_i\cap\mathcal{L}\} \leq  e^{-2E_{ij}^2\sum_{m=1}^{k-1}S_m}$.

		Similarly, if $i\in\mathcal{L}, j\in\mathcal{N}_i\cap\mathcal{M}$, then, $\{T_{ij}=k, i\in\mathcal{L}, j\in\mathcal{N}_i\cap\mathcal{M}\}\subseteq \{\beta_{ij,k-1}\geq0, i\in\mathcal{L}, j\in\mathcal{N}_i\cap\mathcal{M}\}$, which means $\mathbb{P}\{T_{ij}=k, i\in\mathcal{L}, j\in\mathcal{N}_i\cap\mathcal{M}\} \leq \mathbb{P}\{\{\beta_{ij,k-1}\geq0, i\in\mathcal{L}, j\in\mathcal{N}_i\cap\mathcal{M}\} \} \leq  e^{-2E_{ij}^2\sum_{m=1}^{k-1}S_m}$.

    By the definition of $T_f$, the event $\{T_f = k\}$ occurs if and only if there exists at least one edge $(i_0,j_0)$ such that $T_{i_0j_0} = k$ and, for all other edges, $T_{ij} \leq k$. Since the classification processes for different edges are independent, combining the above inequalities gives
		\begin{align}
			\mathbb{P}\{T_{f} = k\} =& \mathbb{P}\{T^{(1)}<k,..., T^{(\sum_{i=1}^{N'} |\mathcal{N}_i|)} = k\}\nonumber\\
			=&\big[\prod_{ij\neq i_0j_0}\mathbb{P}\{T_{ij}\leq k\}\big]\mathbb{P}\{T_{i_0j_0}=k\} \nonumber\\
			\leq& e^{-2E_{i_0j_0}^2\sum_{m=1}^{k-1}S_m} \nonumber\\
			\leq& e^{-2E_{min}^2\sum_{m=1}^{k-1}S_m}. \nonumber
		\end{align}
	\end{proof}
	
	\begin{remark}
        \textcolor{blue}{The results in \cite{yemini2021characterizing} and \cite{yemini2025resilient} depend on the number of players and the uniform classification capability to treat all legitimate players as a single entity with identical trust dynamics. This treatment allows the overall classification error to be bounded by summing over all edges and enables the decoupling of the maximum classification time analysis into two independent error events. However, the heterogeneous identification capabilities prevent us from continuing to handle the problem in this manner. Inspired by their proof methodology, the maximum correct classification time $T_f$ is obtained by focusing on a pair of players and constructing the order statistic to analyze the probability of $T_f$. As a result, the obtained theorem addresses the issue of heterogeneity. Besides, it ensures that the probability of $T_f$ is solely dependent on the pair of players with the weakest classification ability, denoted as $E_{ij} = E_{min}$, and is independent of the total number of players.}
	\end{remark}
	
	\begin{remark}
		Different from \cite{yemini2021characterizing}\cite{yemini2025resilient} which extract trust observations only once per iteration, each player allows for the extraction of trust observations up to $S_k$ times at each iteration. According to Theorem \ref{SectionTrustTheorem1}, the expected time to correctly classify all players is derived as $\sum_{k=2}^{\infty}(k-1) \mathbb{P}\{T_{f}=k\} \leq  \sum_{k=2}^{\infty} (k-1) e^{-2E_{min}^2\sum_{m=1}^{k-1}S_m}$. As a special case that extracts only one trust observation in each iteration, i.e., $S_k = 1, k=1, 2, ..., K $, its expectation is $\sum_{k=2}^{\infty} (k-1) \mathbb{P}\{T_{f}=k\}\leq \sum_{k=2}^{\infty} (k-1)   e^{-2E_{min}^2(k-1)}$.
		Clearly, $(k-1) e^{-2E_{min}^2\sum_{m=1}^{k-1}S_m} \leq (k-1)   e^{-2E_{min}^2(k-1)}$ holds. This indicates that allowing $S_k$ extractions per iteration effectively reduces the upper bound of the expected time for correct classification, which in turn leads to a decrease in the times of gradient computations.
	\end{remark}
	
	\subsection{The trustworthiness broadcast mechanism}\label{FloodingSection}

    Once the malicious players are classified within a random finite time $T_f$, they can be isolated by setting $a_{ij} = 0, i\in\mathcal{L}, j\in\mathcal{N}_i\cap\mathcal{M}$, and the aggregative game becomes a game where only legitimate players participate. \textcolor{blue}{Note that the aggregative game after isolating the malicious players is different from those without the presence of malicious players from the beginning. This is because the initialization at the beginning includes the malicious ones and isolating them cannot ensure a legal initialization of the game from then onwards. More specifically, since $v_{i,T_f} = e_{i},  \sigma_{i,T_f} = \phi_i(x_{i,T_f})$ is not satisfied, $T_f$ cannot be taken as the initial time of Algorithm \ref{Algorithm 2}.}

    \textcolor{blue}{To address this, the trustworthiness broadcast mechanism is designed as follows to ensure that once malicious players are isolated, each legitimate player not only knows the other legitimate players but also resets its variables to the initial inputs required by Algorithm \ref{Algorithm 2}.}
	
	If $\mathcal{N}_{i,k-1} \neq \mathcal{N}_{i,k}$, set $\iota_{i,k} = N'$, otherwise, $\iota_{i,k}$ does not change.
	\begin{align}
		&\mathcal{A}_{i,k}(i) = k, \label{Broadcast2}\\
		&\mathcal{A}_{i,k+1} = max_{j\in\mathcal{N}_{i,k}}{\mathcal{A}_{j,k}}, \label{Broadcast3}\\
		&\iota_{i,k} = max_{j\in\mathcal{N}_{i,k}}{\iota_{j,k}}, \label{Broadcast4}\\
		&list_{i,k} = \{j|\mathcal{A}_{i,k+1}(i)-\mathcal{A}_{i,k+1}(j)\leq N'-1\}, \label{Broadcast5}\\
		&i' = find(list_{i,k}==i) \label{Broadcast6}. \\
		&\iota_{i,k+1} = \iota_{i,k} - 1,\quad and \quad reset \quad v_{i,k}, \sigma_{i,k}, \quad if\quad\iota_{i,k} \neq 0. \label{Awareness} 
	\end{align}
	
	The core of the broadcast mechanism is that any player in the network can send information to another player with at most $N'-1$ information transmissions. The following provides an intuitive explanation of how players can become aware of the presence of malicious players, and how they can obtain a list of legitimate players.
	
	\textbf{1)} After a change in player $i$'s classification of neighbors, it takes at most $N'-1$ iterations for all players to be informed. $\iota_{i,k}$
	serves as a counter, indicating the maximum number of additional iterations needed for player $i$'s information to reach the entire network. When the counter has not reached $0$, player $i$ resets $v_{i,k}, \sigma_{i,k}$ to initial values and decrements the counter by $1$. Taking Fig. \ref{fig:communication-network5} as an example, consider a scenario where player $1$ updates its classification at iteration $k$ to classify player $4$ as malicious. Player $1$ sets the counter $\iota_{1, k}=5$, and sends it to player $2$, resulting in $\iota_{2, k}=5$ by \eqref{Broadcast4}. Subsequently, we have $\iota_{1, k+1}=\iota_{2, k+1}=4$ by \eqref{Awareness}. Similarly, after three rounds, all legitimate players have synchronized their counters $\iota_{1, k+3}=\iota_{2, k+3}=\iota_{3, k+3}=\iota_{5, k+3}=2$, indicating that they are aware of the presence of malicious players.
	
	\textbf{2)} After $T_f$, when malicious players are isolated, player $i$ requires only an $N'$-dimensional vector $\mathcal{A}_{i,k}$, where each element $\mathcal{A}_{i,k}(j)$ marks the latest iteration count of player $j$. According to \eqref{Broadcast3}, the message that 'player $j$ has not yet been isolated at iteration $k$, i.e., $\mathcal{A}_{j,k}(j)=k$' will be known to $i$ by no later than iteration $k+N'-1$. Since malicious players are isolated before $T_f$, their highest mark is $T_f$. Thus, if another player's iteration count differs from player $i$'s by more than $N'-1$ after iteration $T_f+N'$, that player is considered malicious. Then, we can obtain a list of legitimate players $list_{i,k}$, and determine player $i$'s number $i'$ among them. For example, at $T_f+5$, if player $5$'s vector $\mathcal{A}_{5,T_f+5}$ still shows player $4$'s iteration count $\mathcal{A}_{5,T_f+5}(4)=T_f$ , then player $5$ can infer that player $4$ is malicious. Otherwise, player $5$ would have received an iteration count from its neighbors indicating that player $4$ participated in the game at least by $T_f+1$. At this point, player $5$ can obtain a list of legitimate players $list_{5,T_f+5} = \{1,2,3,5\}$, and its number among legitimate players is $4$.
	\begin{figure}
		\centering
		\includegraphics[width=0.7\linewidth]{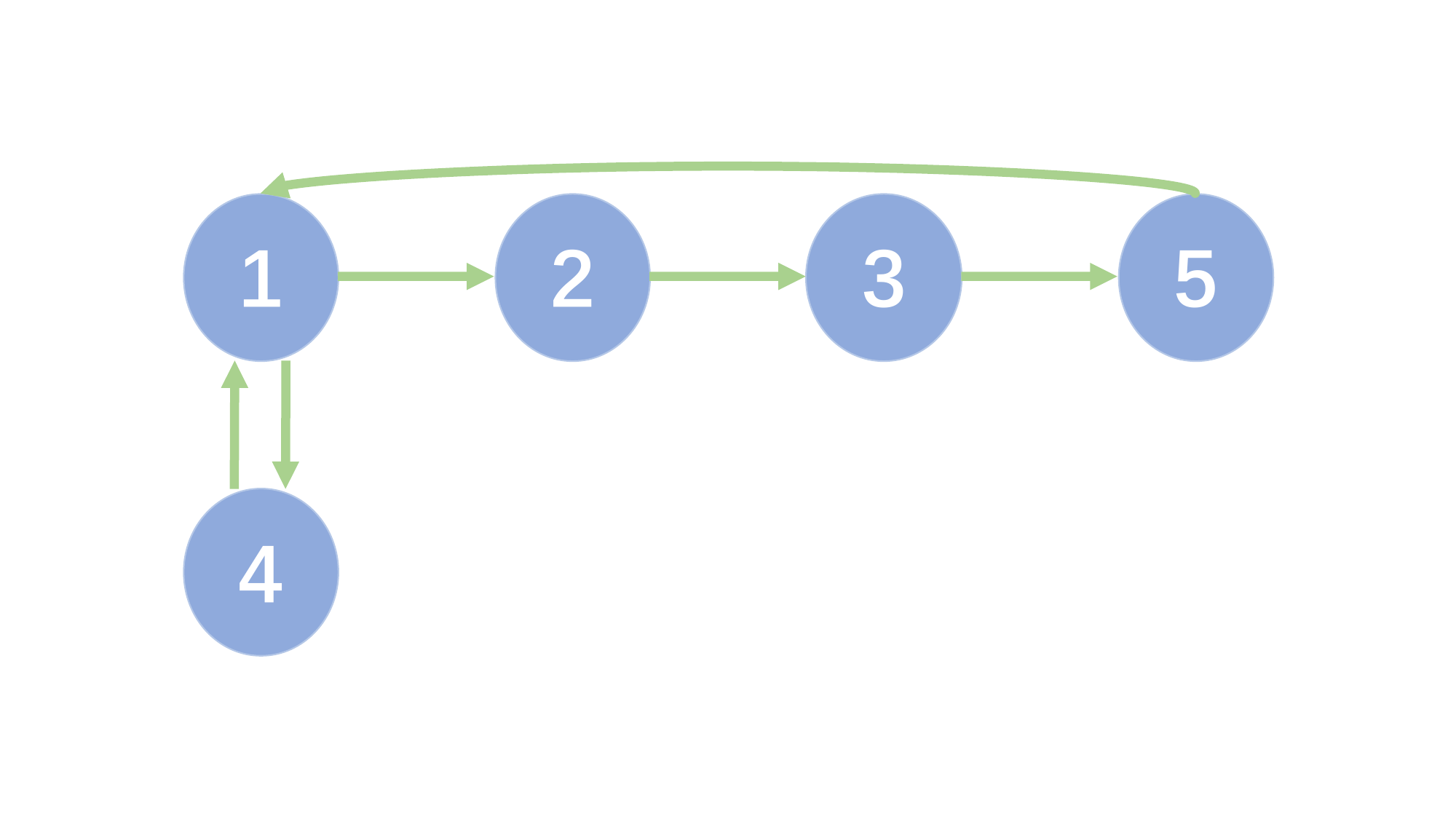}
		\caption{The directed communication network with legitimate players $1,2,3,5$ and a malicious player $4$.}
		\label{fig:communication-network5}
	\end{figure}

	In the following it will prove that after time $T_f$, \textbf{1)} the counter $\iota_{i,k}$ of legitimate players will simultaneously become $0$ at $T_f+N'$; \textbf{2)} each legitimate player will obtain a list of all legitimate players, i.e. $list_{i,k} = \mathcal{L}, \forall i\in\mathcal{L}, k\geq T_f+N'$. 

	\textbf{1)} Assuming that at iteration $T_{f}$, players $i_0^1, i_0^2, ..., i_0^{n_0}$ change their classification, denote them as a set $\mathcal{C}_{0} = \{i_0^1, i_0^2, ..., i_0^{n_0}\}$. Due to the definition of $T_f$, there is no more change in the classification of legitimate players after $T_f$. The following  analysis is divided into three stages: initiation $T_f+1$, propagation during $T_f+2$ to $T_f+N'-1$, and termination at $T_f+N'$. 
	
	\textcircled{\scriptsize{1}}Since $\mathcal{N}_{i_0, T_f-1} \neq \mathcal{N}_{i_0, T_f}$, $\iota_{i_0, T_f} = N', \forall i_0 \in \mathcal{C}_0$. For $i_1$ that satisfies $\exists i_0 \in \mathcal{C}_0, i_0 \in \mathcal{N}_{i_1, T_f}$, there is $\iota_{i_1,T_{f}} = \max_{j\in\mathcal{N}_{i_1,T_f}} \iota_{j,T_{f}} = \iota_{i_0,T_{f}} =N'$ by \eqref{Broadcast4}. Denote $\mathcal{C}_1 = \{i_1|\exists i_0 \in \mathcal{C}_0, i_0 \in \mathcal{N}_{i_1, T_f}\}$. It is obvious that $\mathcal{C}_0 \subseteq \mathcal{C}_1$. After that, for $i\in \mathcal{C}_1$, because $\iota_{i,T_{f}} \neq 0$, there is $\iota_{i,T_f+1}= \iota_{i,T_f}- 1 = N'-1$ by \eqref{Awareness}. 
	
	\textcircled{\scriptsize{2}}Following similar steps, at $T_{f}+r$, $\iota_{i, T_f+r} = N'-r, \forall i\in \mathcal{C}_r$, and $\iota_{i, T_f+r} \leq N'-r, \forall i \notin \mathcal{C}_r$. Then, for $i_{r+1}$ that satisfies $\exists i_r \in \mathcal{C}_r, i_r \in \mathcal{N}_{i_{r+1}, T_f+r}$, there is $\iota_{i_{r+1},T_{f}+r} = \max_{j\in\mathcal{N}_{i_{r+1},T_f+r}} \iota_{j,T_{f}+r} =N'-r$ by \eqref{Broadcast4}. Denote $\mathcal{C}_{r+1} = \{i_{r+1}|\exists i_r \in \mathcal{C}_r, i_{r} \in \mathcal{N}_{i_{r+1},T_f+r}\}$. After that, for $i\in \mathcal{C}_{r+1}$, because of $\iota_{i,T_{f}+r} \neq 0$, there is $\iota_{i,T_f+r+1} = N'-r-1$ by \eqref{Awareness}. 
	
	\textcircled{\scriptsize{3}}At $T_{f}+N'-1$, for $i_{N'}$ that satisfies $\exists i_{N'-1} \in \mathcal{C}_{N'-1}, i_{N'-1} \in \mathcal{N}_{i_{N'}, T_f+N'-1}$, there is $\iota_{i_{N'},T_{f}+N'-1} = 1$ by \eqref{Broadcast4}. Denote $\mathcal{C}_{N'} = \{i_{N'}|\exists i_{N'-1} \in \mathcal{C}_{N'-1}, i_{N'-1} \in \mathcal{N}_{i_N', T_f+N'-1}\}$. After that, for $i\in \mathcal{C}_{N'}$, because $\iota_{i,T_{f}+N'-1} \neq 0$, there is $\iota_{i,T_f+N'} = 0$ by \eqref{Awareness}. 
	
	Due to Assumption \ref{AssumptionTrustValue}-1), there is $\mathcal{L} \subseteq \mathcal{C}_{N'}$. This indicates that the counters of all legitimate players have simultaneously become $0$ at $T_f+N'$.
	
	Take Fig. \ref{fig:communication-network5} as an example, player $1$ updates its classification at iteration $k$ to classify player $4$ as malicious, then player $1$ is in $\mathcal{C}_{0}$. The updates of the counters are listed in Table \ref{Example} ($N' = 5$). Record $\iota_{i, k}$ in a column for iteration $k$, i.e., $\iota_{2, T_f+1} = N'-1$.
	
	\begin{table}[H]
		\small
		\caption{The counter $\iota_{i, k}$ of all players after $T_f$}
		\label{Example}
		\begin{minipage}[c]{0.5\linewidth}
			\begin{tabular}{|c|c|c|c|c|c|c|}
				\hline
				$\mathcal{C}_r$ & players & $k=T_f$ & $T_f+1$ & $T_f+2$ & $T_f+3$ & $T_f+4$ \\
				\hline
				$\mathcal{C}_0$ & 1 & $N'$ & $N'-1$ & $N'-2$ & $N'-3$ & $N'-4$ \\
				\hline
				$\mathcal{C}_1$ & 1;2 &  & $N'-1$ & $N'-2$ & $N'-3$ & $N'-4$   \\
				\hline
				$\mathcal{C}_2$ & 1,2;3 &  & & $N'-2$ & $N'-3$ & $N'-4$   \\
				\hline
				$\mathcal{C}_3$ & 1,2,3;5 &  & & & $N'-3$ & $N'-4$   \\
				\hline
			\end{tabular}
		\end{minipage}
	\end{table}
	
	\textbf{2)} From \eqref{Broadcast2} and \eqref{Broadcast3}, there is $ \mathcal{A}_{i,k}(j)<T_f, \forall i\in\mathcal{L}, j\in\mathcal{M}, k\geq T_f$, and $\mathcal{A}_{i,k}(j)\geq k-N', \forall i\in\mathcal{L}, j\in\mathcal{L}, k\geq T_f+N'$. If not, there would exist $ i_0\in\mathcal{L}, j_0 \in \mathcal{M}, k_0 \geq T_f$, such that $\mathcal{A}_{i_0, k_0}(j_0)\geq T_f$. By \eqref{Broadcast3}, there would exist $ i_1 \in \mathcal{N}_{i_0,k_0-1}$, such that $\mathcal{A}_{i_1, k_0-1}(j_0)\geq T_f$. Through a similar derivation, it can be deduced that there would exist $ i_{k_0-T_f} \in \mathcal{N}_{i_{k_0-T_f-1},T_f}$, such that $\mathcal{A}_{i_{k_0-T_f}, T_f}(j_0)\geq T_f$. From \eqref{Broadcast2}, it follows that $j_0$ is $i_{k_0-T_f}$, then, $j_0 \in\mathcal{N}_{i_{k_0-T_f-1},T_f}$ and $i_0, i_{1}, ..., i_{k_0-T_f-1}$ are legitimate players. However, this contradicts the definition of $T_f$. Similarly, $\mathcal{A}_{i,k}(j)\geq k-N', \forall i\in\mathcal{L}, j\in\mathcal{L}, k\geq T_f+N'$.
	
	To explain $list_{i,k} = \mathcal{L}, k\geq T_f+N'$, suppose there exists $ j_0 \in \mathcal{M}, i_0\in\mathcal{L}$ such that $j_0 \in list_{i_0,k}$, there is $\mathcal{A}_{i_0, k+1}(i_0) -  \mathcal{A}_{i_0, k+1}(j_0) \leq N'-1$. Consequently, $\mathcal{A}_{i_0, k+1}(j_0)\geq k-N'+1 \geq T_f+N'-N'+1 = T_f +1$. However, it contradicts the previously conclusion $ \mathcal{A}_{i,k}(j)<T_f, \forall i\in\mathcal{L}, j\in\mathcal{M}, k\geq T_f$. Similarly, it is also impossible for there to exist $j_0 \in \mathcal{L}, i_0\in\mathcal{L}$ such that $j_0 \notin list_{i_0,k}$. 
	
	After $T_f+N'$, players can seek the NE using Algorithm \ref{Algorithm 2}. Since every legitimate player knows the list of all legitimate players, it can determine its number among the legitimate players using \eqref{Broadcast6}. Then, the weight matrix of all players $A'$ after $T_{f}+N'$ can be written as
	\begin{align}\nonumber
		A' = \begin{bmatrix}
			A & 0_{|\mathcal{L}|\times|\mathcal{M}|}\\
			B & C \\
		\end{bmatrix}, k\geq T_{f}+N',
	\end{align}
	where the weights between legitimate players (the weight that player $i$ assigns to player $j$ by \eqref{weights}) are written in the first $N$ rows and the first $N$ columns. $A\in \mathbb{R}^{|\mathcal{L}|\times|\mathcal{L}|}$, $B\in \mathbb{R}^{|\mathcal{M}|\times|\mathcal{L}|}$, and $C\in \mathbb{R}^{|\mathcal{M}|\times|\mathcal{M}|}$. We write $v_{i,k}$ as $[v_{i,k}^\mathcal{L}, v_{i,k}^\mathcal{M}]$, where $v_{i,k}^\mathcal{L} \in \mathbb{R}^{1\times|\mathcal{L}|}$ and $v_{i,k}^\mathcal{M} \in \mathbb{R}^{1\times|\mathcal{M}|}$. Thus, \eqref{Multi-eigenvalue} can be written as
	\begin{align}\nonumber
		&\begin{bmatrix}
			v_{1,k,l+1}^\mathcal{L}, v_{1,k,l+1}^\mathcal{M} \\
			...\\
			v_{N,k,l+1}^\mathcal{L}, v_{N,k,l+1}^\mathcal{M} \\
			v_{N+1,k,l+1}^\mathcal{L}, v_{N+1,k,l+1}^\mathcal{M}\\
			...\\
			v_{N',k,l+1}^\mathcal{L}, v_{N',k,l+1}^\mathcal{M}
		\end{bmatrix} = 
		A'
		\begin{bmatrix}
			v_{1,k,l}^\mathcal{L}, v_{1,k,l}^\mathcal{M} \\
			...\\
			v_{N,k,l}^\mathcal{L}, v_{N,k,l}^\mathcal{M} \\
			v_{N+1,k,l}^\mathcal{L}, v_{N+1,k,l}^\mathcal{M}\\
			...\\
			v_{N',k,l}^\mathcal{L}, v_{N',k,l}^\mathcal{M}
		\end{bmatrix},\nonumber\\
		&k \geq T_{f}+N', \nonumber
	\end{align}
	which means
	\begin{align}\label{Connection1}
		\begin{bmatrix}
			v_{1,k,l+1}^\mathcal{L} \\
			...\\
			v_{N,k,l+1}^\mathcal{L} \\
			
		\end{bmatrix} 
		= 
		A 
		\begin{bmatrix}
			v_{1,k,l}^\mathcal{L} \\
			...\\
			v_{N,k,l}^\mathcal{L}\\
		\end{bmatrix},  k \geq T_{f}+N'
	\end{align}
	and $v_{i,k,l}^\mathcal{M} = 0, \forall i\in\{1,2,...,N\}$ (notice that $v_{i,T_{f}+N'} = e_{i'}$). Similarly, \eqref{Multi-estimate} can be written as
	\begin{align}\label{Connection2}
		\begin{bmatrix}
			\sigma_{1,k,l+1}^\mathcal{L} \\
			...\\
			\sigma_{N,k,l+1}^\mathcal{L} \\
			
		\end{bmatrix} 
		= 
		A 
		\begin{bmatrix}
			\sigma_{1,k,l}^\mathcal{L} \\
			...\\
			\sigma_{N,k,l}^\mathcal{L}\\
		\end{bmatrix},  k \geq T_{f}+N'.
	\end{align}
	Note that $v_{i,k,l}^\mathcal{L}$, $\sigma_{i,k,l}^\mathcal{L}$ correspond to $v_{i,k,l}$, $\sigma_{i,k,l}$ in \eqref{Multi-eigenvalue}\eqref{Multi-estimate}, respectively. Therefore, after $T_f+N'$, the iterations of legitimate players are precisely described by \eqref{AlgorithmUpdate2}\eqref{AlgorithmUpdate1}\eqref{AlgorithmUpdate3} with initial inputs $v_{i,T_f+N'} = e_i,  \sigma_{i,T_f+N'} = \phi_i(x_{i,T_f+N'})$.

	\subsection{Convergence Analysis}\label{ConvergenceSection}
	This subsection analysis the convergence of the resilient fast NE seeking algorithm, which integrates the heterogeneous trustworthiness probabilistic framework with the fast NE seeking algorithm \ref{Algorithm 2} through the trustworthiness broadcast mechanism, as presented in Algorithm \ref{Algorithm 1}.

	\begin{algorithm}
		\caption{The resilient fast NE seeking algorithm}
		\label{Algorithm 1}
		\small
		\begin{algorithmic}
			\STATE\textbf{Input}: $ x_{i,1} \in \Omega_i, v_{i,1} = e_{i'},  \sigma_{i,1} = \phi_i(x_{i,1})$, $\beta_{ij,k} = 0, \forall j\in \mathcal{N}_i$, $\mathcal{N}_{i,0} = \mathcal{N}_i$, $\iota_{i,1} = 0$, $\mathcal{A}_{i,k} = 0_{N'\times 1}$.
			\STATE\textbf{Output}: Players' strategies $x_i, \forall i \in\mathcal{N}$.
			\FOR {$k = 1, 2, ..., K$}
			\FOR {each player $i \in \mathcal{N}$}
			\STATE Determine the trusted in-neighbors $\mathcal{N}_{i,k} =\{j\in\mathcal{N}_i:\beta_{ij,k}\geq0\}$ and neighbors' weights as \eqref{weights}.
			\STATE \textbf{Perform} the \textbf{\textit{trustworthiness broadcast mechanism}}
			If $\mathcal{N}_{i,k-1} \neq \mathcal{N}_{i,k}$, set $\iota_{i,k} = N'$, otherwise, $\iota_{i,k}$ does not change. \eqref{Broadcast2}-\eqref{Broadcast6}.
			\STATE \textbf{Perform} the \textbf{\textit{Multi-Round Communication Mechanism } } Communicate with all neighbors $\mathcal{N}_i$ as Multi-round Communication Mechanism
			\begin{align}
				&\sigma_{i,k,S_k}, v_{i,k,S_k} \leftarrow MRC (\sigma_{i,k}, v_{i,k}), \nonumber
			\end{align}
			\STATE \textbf{Updates} trustworthiness by extracting $\tau_{ij,k,l}$
			\begin{align}\label{AlgorithmIdentify1} 
				&\beta_{ij,k+1} = \beta_{ij,k} + \sum_{l=0}^{S_k-1} (\tau_{ij,k,l}-\frac{1}{2}). \nonumber
			\end{align}
			\IF{$\iota_{i,k} = 0$}
			\STATE Player $i$ updates its variables
			\begin{align}
				& v_{i,k+1} = v_{i,k,S_k}, \nonumber\\
				&x_{i,k+1} =P_{\Omega_i} [x_{i,k} - \alpha_k G_i(x_{i,k},\sigma_{i,k})], \nonumber\\
				&\sigma_{i,k+1} = \sigma_{i,k,S_k} + \frac{\phi_i(x_{i,k+1})}{[v_{i,k+1}]_{i'}} - \frac{\phi_i(x_{i,k})}{[v_{i,k}]_{i'}},\nonumber
			\end{align}
			\STATE \textbf{return} $v_{i,k+1}, x_{i,k+1}, \sigma_{i,k+1}$
			\ELSE
			\STATE Resets its variables as
			\begin{align}
				& v_{i,k+1} = e_{i'}; 
				x_{i,k+1} =x_{i,1}; 
				\sigma_{i,k+1} = \sigma_{i,1},\nonumber\\
				&\iota_{i,k+1} = \iota_{i,k} - 1 .\nonumber
			\end{align}
			\ENDIF
			\ENDFOR
			\ENDFOR
		\end{algorithmic}
	\end{algorithm}
	Player $i$ first sends and receives information with neighbors, including the trustworthiness broadcast mechanism \eqref{Broadcast2}-\eqref{Broadcast6} and multi-round communication mechanism \eqref{Multi-estimate}\eqref{Multi-eigenvalue}. $\tau_{ij,k,l}$ are extracted at the same time and $\beta_{ij,k}$ is then updated; if the counter $\iota_{i,k}$ is $0$, the $v_{i,k+1}, x_{i,k+1}, \sigma_{i,k+1}$ are updated by \eqref{gradientbasedupdate1}-\eqref{gradientbasedupdate3}, otherwise, the variables are reset to their initial values. \textcolor{blue}{Note that, the index of player $i$ within the set of legitimate players is denoted as $i'$. Owing to this initialization and weight estimation \eqref{gradientbasedupdate1}, the estimated weight correspond to the legitimate player $i$ is $i'$-th element of $v_{i,k}$. When only legitimate players are involved, $i' = i$. However, this equality does not hold in the earlier stages of the overall algorithm, where the legitimate set may still include misclassified players or exclude some legitimate ones. Therefore, we explicitly distinguish $i'$ from $i$ in our notation.} 
	
	 We will now proceed to study the convergence of Algorithm \ref{Algorithm 1}.
	 \begin{theoremx} \label{SectionTheorem3}
	 	Suppose that Assumptions \ref{AssumptionGameSets}, \ref{AssumptionGameGradient}, \ref{AssumptionLips}, \ref{AssumptionCommunicationGraph}, and \ref{AssumptionTrustValue} hold. 
	 	\begin{enumerate}[1)]
	 		\item Let step-size $\alpha_k = \alpha, \forall k\in\mathbb{N}_+$ such that $\rho(\Gamma_{\alpha}) = \lambda_{max}(\Gamma_{\alpha})<1$,
	 		then the sequence $\{x_k\}_{k\in\mathbb{N}_+}$ generated by Algorithm \ref{Algorithm 1} satisfies $\lim\limits_{k\to\infty} \mathbb{E}\|x_k-x^*\|^2 \leq\frac{D_3e^{2E_{min}^2}}{(1-\rho(\Gamma_{\alpha}))(1-e^{-2E_{min}^2})}$.
	 		\item Let step-size $\alpha_k = \frac{1}{(k+1)^r}, \forall k\in\mathbb{N}_+$, where $\frac{1}{2}<r<1$. Let $S_k \geq \frac{\ln\frac{1}{2}}{2\ln\rho_1}$. Then, the sequence $\{x_k\}_{k\in\mathbb{N}_+}$ generated by Algorithm \ref{Algorithm 1} satisfies $\lim\limits_{k\to\infty} \mathbb{E}\|x_k-x^*\|^2 = 0$.
	 	\end{enumerate}
	 \end{theoremx}
	 \begin{proof}
	 	Denote $\|x_k^{\mathcal{L}}-x^*\|^2$ as the squared error of Algorithm \ref{Algorithm 2} after correctly classifying all neighbors. Assume that malicious players exists, therefore $T_f\geq 2$.
	 	It is obvious that
	 	\begin{align}
	 		&\mathbb{P}\{T_{f}+N' >k\} = \sum_{l=k+1}^{\infty} \mathbb{P}\{T_{f} =l-N'\} \nonumber\\
	 		\leq& \sum_{l=k+1}^{\infty} e^{-2E_{min}^2\sum_{m=1}^{l-N'-1}S_m} 
	 		\leq \frac{e^{-2E_{min}^2(k-N')}}{1-e^{-2E_{min}^2}}, \nonumber
	 	\end{align}
	 	and by using the law of total expectation, it has
	 	\begin{align}
	 		&\mathbb{E}\|x_k-x^*\|^2 \nonumber\\
	 		\leq& \sum_{l=N'+2}^{k}\mathbb{P}\{T_{f} = l-N'\}\mathbb{E} \{\|x_{k-l}^\mathcal{L}-x^*\|^2|T_{f} = l-N'\}\nonumber\\
	 		& + \mathbb{P}\{T_{f} >k-N'\}C.
	 	\end{align}
	 	\textbf{1)}: For step-size $\alpha_k = \alpha, \forall k\in\mathbb{N}_+$ with $\rho(\Gamma_\alpha) = \lambda_{max}(\Gamma_{\alpha})<1$, by using Theorem \ref{SectionGameTheorem1},
	 	\begin{align}
	 		&\mathbb{E}\|x_k-x^*\|^2 \nonumber\\
	 		\leq& \sum_{l=N'+2}^{k}\mathbb{P}\{T_{f} = l-N'\}\mathbb{E} \{\|x_{k-l}^\mathcal{L}-x^*\|^2|T_{f} = l-N'\}\nonumber\\
	 		& + \mathbb{P}\{T_{f} > k-N'\}C \nonumber\\
	 		\leq& \sum_{l=N'+2}^{k} e^{-2E_{min}^2(l-N'-1)} \cdot \nonumber\\ &\bigg[\rho(\Gamma_{\alpha})^{k-l}(\|x_1-x^*\|^2 + \|\sigma_1 - V_{\infty}\sigma_1\|^2)  \nonumber\\
	 		&+D_1\frac{\theta^2[\rho(\Gamma_\alpha)^{k-l} - (\theta^2)^{k-l}]}{\rho(\Gamma_\alpha) - \theta^2} + D_2\frac{\theta[\rho(\Gamma_\alpha)^{k-l} - (\theta)^{k-l}]}{\rho(\Gamma_\alpha) - \theta} \nonumber\\
	 		&+ D_3\frac{1-\rho(\Gamma_\alpha)^{k-l}}{1-\rho(\Gamma_\alpha)} \bigg] + \frac{e^{-2E_{min}^2(k-N')}}{1-e^{-2E_{min}^2}} C. \nonumber
	    \end{align}
	 	Since, for example, 
	 	\begin{align}\label{SectionConvergence1}
	 		&\lim\limits_{k\to\infty}\sum_{l=N'+2}^{k} e^{-2E_{min}^2(l-N'-1)} \cdot \rho(\Gamma_{\alpha})^{k-l} \nonumber\\ =&\lim\limits_{k\to\infty}\frac{e^{2E_{min}^2}}{\rho(\Gamma_{\alpha})^{N'-1}} \frac{\rho(\Gamma_{\alpha})^k-e^{-2E_{min}^2k}}{\rho(\Gamma_{\alpha})-e^{-2E_{min}^2}} = 0,
	 	\end{align}
	 	 $\lim\limits_{k\to\infty}\sum_{l=N'}^{k} e^{-2E_{min}^2(l-N'-1)}  (\theta^2)^{k-l}  = 0$, and $\lim\limits_{k\to\infty}\sum_{l=N'}^{k} e^{-2E_{min}^2(l-N'-1)}  \theta^{k-l}  = 0$, it has
	 	\begin{align}
	 		\lim\limits_{k\to\infty} \mathbb{E}\|x_k-x^*\|^2 \leq\frac{D_3e^{2E_{min}^2}}{(1-\rho(\Gamma_{\alpha}))(1-e^{-2E_{min}^2})}, \nonumber
	 	\end{align}
	 	where $D_3 = 2\rho_1^{S_k}\tilde{M}\sqrt{[v]_{max}}\rho_2\rho L_3M\alpha+2[v]_{max}\rho_2^2\rho^2 L_3^2M^2\alpha^2$. 
	 	
	 	\textbf{2)}: Let step-size $\alpha_k = \frac{1}{(k+1)^r}, \forall k\in\mathbb{N}_+$, where $\frac{1}{2}<r<1$. Let $S_k \geq \frac{\ln\frac{1}{2}}{2\ln\rho_1}$.
	 	\begin{align}
	 		&\mathbb{E}\|x_k-x^*\|^2 \nonumber\\
	 		\leq& \sum_{l=N'+2}^{k}\mathbb{P}\{T_{f} = l-N'\}\mathbb{E} \{\|x_{k-l}^\mathcal{L}-x^*\|^2|T_{f} = l-N'\} \nonumber\\
	 		&+ \mathbb{P}\{T_{f} = l-N' >k\}C\nonumber\\
	 		\leq& \sum_{l=N'+2}^{k} e^{-2E_{min}^2(l-N'-1)} \cdot \nonumber\\
	 		&\bigg[\big(\|x_1-x^*\|^2 + \|\sigma_1 - V_{\infty}\sigma_1\|^2\big)\frac{1}{k-l+1} \nonumber\\
	 		&+d_1\big[\frac{1}{k-l+1}\frac{\theta^2}{1-\theta^2} + \frac{1}{k-l+1}\frac{\theta^2(1-(\theta^2)^{k-l})}{(1-\theta^2)^2}\nonumber\\
	 		& - \frac{(\theta^2)^{k-l+1}}{1-\theta^2}\big] \nonumber\\
	 		&+ d_2\big[\frac{1}{k-l+1}\frac{\theta}{1-\theta} + \frac{1}{k-l+1}\frac{\theta(1-\theta^{k-l})}{(1-\theta)^2} - \frac{\theta^{k-l+1}}{1-\theta}\big] \nonumber\\
	 		&+ d_3\frac{1}{2-2r}\big[\frac{1}{(k-l+1)^{2r-1}} - \frac{1}{k-l+1}\big]\bigg] \nonumber\\
	 		&+  \frac{e^{-2E_{min}^2(k-N')}}{1-e^{-2E_{min}^2}} C. \nonumber
	 	\end{align}
	 	There is
	 	\begin{align}
	 		&\sum_{l=N'+2}^{k}e^{-2E_{min}^2(l-N'-1)}\frac{1}{(k-l+1)^{2r-1}} \nonumber\\
	 		=& e^{2E_{min}^2N'}\sum_{l=N'+1}^{k-1}e^{-2E_{min}^2l}\frac{1}{(k-l)^{2r-1}} \nonumber\\
	 		\leq & e^{2E_{min}^2N'}\big[\frac{1}{k^{2r-1}} + \sum_{l=1}^{k-2}e^{-2E_{min}^2l}\frac{1}{(k-l)^{2r-1}} \nonumber\\
	 		& + e^{-2E_{min}^2(k-1)}\big],
	 	\end{align}
		define $f(x) \triangleq e^{-2E_{min}^2x}\frac{1}{(k-x)^{2r-1}}$, $f(x)$ decreases when $x\in(0,k-\frac{2r-1}{2E_{min}^2})$, and increases when $x\in(k-\frac{2r-1}{2E_{min}^2},k)$. Denote $k_0 = \lceil k-\frac{2r-1}{2E_{min}^2} \rceil -1$, then,
	 	\begin{align}
	 		&e^{-2E_{min}^2l}\frac{1}{(k-l)^{2r-1}} \nonumber\\
	 		\leq& \int_{l-1}^{l} e^{-2E_{min}^2x}\frac{1}{(k-x)^{2r-1}} dx, l=\{1, 2, ..., k_0\}, \nonumber\\
	 		&e^{-2E_{min}^2l}\frac{1}{(k-l)^{2r-1}} \nonumber\\
	 		\leq& \int_{l}^{l+1} e^{-2E_{min}^2x}\frac{1}{(k-x)^{2r-1}} dx, l=\{k_0+1, ..., k-2\}.
	 	\end{align}
	 	Thus, 
	 	\begin{align}
	 		\sum_{l=1}^{k-2}e^{-2E_{min}^2l}\frac{1}{(k-l)^{2r-1}} \leq \int_{0}^{k-1} e^{-2E_{min}^2x}\frac{1}{(k-x)^{2r-1}} dx, \nonumber
	 	\end{align}
	 	and $\lim\limits_{k\to\infty}\int_{0}^{k-1} e^{-2E_{min}^2x}\frac{1}{(k-x)^{2r-1}} dx=0$. Then, $\lim\limits_{k\to\infty}\sum_{l=1}^{k}e^{-2E_{min}^2(l-1)}\frac{1}{(k-l+1)^{2r-1}} = 0$, which means $\lim\limits_{k\to\infty}\sum_{l=N'+2}^{k}e^{-2E_{min}^2(l-N'-1)}\frac{1}{(k-l+1)^{2r-1}}=0$. 
	 	
	 	Similarly, $\lim\limits_{k\to\infty}\sum_{l=N'}^{k}e^{-2E_{min}^2(l-N'-1)}\frac{1}{(k-l+1)} = 0$ and $\lim\limits_{k\to\infty}\sum_{l=N'}^{k}e^{-2E_{min}^2(l-N'-1)} \theta^{k-l+1} = 0$.
	 	Therefore, it can be concluded $\lim\limits_{k\to\infty} \mathbb{E}\|x_k-x^*\|^2 = 0$.
	 \end{proof}

  \begin{remark}
        \textcolor{blue}{Recent advances have also been made in distributed optimization under malicious environments. \cite{han2025byzantine} contributes to distributed optimization in the presence of Byzantine agents, but still requires constraints on both the number and weights of malicious neighbors for each node. \cite{akgun2025learning} adopts the homogeneous trustworthiness probabilistic framework and designs a trust mechanism for distributed optimization; however, the homogeneity limiting its applicability. The proposed heterogeneous trustworthiness probabilistic framework can avoid these assumptions and provides a new perspective for handling malicious environments.}
    \end{remark}

	\section{Simulations}
	This section considers an energy consumption game involving eight plug-in hybrid electric vehicles to verify the results. Players' cost functions follow the design in \cite{02Fang2022Directed} \cite{simpson2006cost}\cite{wen2020distributed}:
	\begin{align} \nonumber
		f_i(x_i,\sigma(x))=x_i(a\sum\limits_{i=1}^Nx_i+b)+p_0\left(\kappa_i(1-\frac{x_i}{x_i^{max}})^2+d_i\right),
	\end{align}
	where $x_i$ is player $i$'s electricity consumption, and $[x_i^{min}, x_i^{max}]$ represents the capacity of player $i$'s battery. $a\sum\limits_{i=1}^Nx_i+b$ is the electricity price that depends on total electricity consumption. $p_0$ is the price of the fuel.  $\kappa_i\left(1-\frac{x_i}{x_i^{max}}\right)^2+d_i$ is the demand of fuel where $d_i$ is player $i$'s basic demand of fuel and $\kappa_i$ is $i$'s conversion factor of fuel demand. The parameters are listed in Table \ref{Tabel}.
	\begin{table}[H]\small
			\caption{PARAMETERS OF PHEVS IN THE SIMULATION}
			\label{Tabel}
			\centering
			\setlength{\tabcolsep}{1mm}{
				\begin{tabular}{c|c|c|c|c|c|c|c|c}
					\hline
					Player $i$ & Type & $x_i^{min}$ & $x_i^{max}$ & $p_0$ & $a$ & $b$ & $\kappa_i$ & $d_i$\\
					\hline
					1 & PHEV1 & 0 & 27.5 & 5.6 & 0.004 & 0.1 & 4.6 & 0.7\\
					2 & PHEV1 & 0 & 27.5 & 5.6 & 0.004 & 0.1 & 4.6 & 0.7\\
					3 & PHEV2 & 0 & 34.2 & 5.6 & 0.004 & 0.1 & 3.7 & 0.8\\
					4 & PHEV2 & 0 & 40.6 & 5.6 & 0.004 & 0.1 & 3.7 & 0.8\\
					5 & PHEV3 & 0 & 40.6 & 5.6 & 0.004 & 0.1 & 3.4 & 0.6\\
					6 & PHEV3 & 0 & 40.6 & 5.6 & 0.004 & 0.1 & 3.4 & 0.6\\
					7 & PHEV4 & 0 & 28.8 & 5.6 & 0.004 & 0.1 & 4.0 & 0.7\\
					8 & PHEV5 & 0 & 35.7 & 5.6 & 0.004 & 0.1 & 3.9 & 0.8\\
					\hline
			\end{tabular}}
	\end{table}
	\textcolor{blue}{Each player seeks to minimize its own charging cost, and the desired stable state is the Nash equilibrium, where no player can benefit from a unilateral deviation.} \textcolor{blue}{This scenario involves heterogeneous players, where each vehicle has its own cost due to differences in design and constraints.} The work in \cite{02Fang2022Directed} provides an example in which the PHEV system is also a potential game, with a potential function $F$ whose Hessian matrix $\nabla^2 F(x)$ is positive definite. Hence, from \cite{bertsekas2003convex}, the example satisfies Assumption \ref{AssumptionGameGradient}.
	The Nash Equilibrium is $x^*=[17.6859, 17.6859, 15.5682, 15.5682, 12.5484, 12.5484, 16.508\\7, 16.3481]$. 
	
	\subsection{Comparison for the fast NE seeking algorithm}
	Given that there is currently no literature exploring the problem of aggregative game equilibrium seeking within the trustworthiness probabilistic framework, this comparison focuses solely on the advantages of Algorithm \ref{Algorithm 2} over existing literature. Since Algorithm \ref{Algorithm 2}, like the one in \cite{02Fang2022Directed}, is designed for unbalanced directed networks, Algorithm \ref{Algorithm 2} is compared with the one in \cite{02Fang2022Directed}. Following the comparative simulations, Theorem \ref{SectionGameTheorem1} and Corollary \ref{SectionGameCorollary1} will be verified. The directed communication graph is illustrated in Fig. \ref{fig:communication8}.
	
	\begin{figure}[H]
		\centering
		\includegraphics[width=0.55\linewidth]{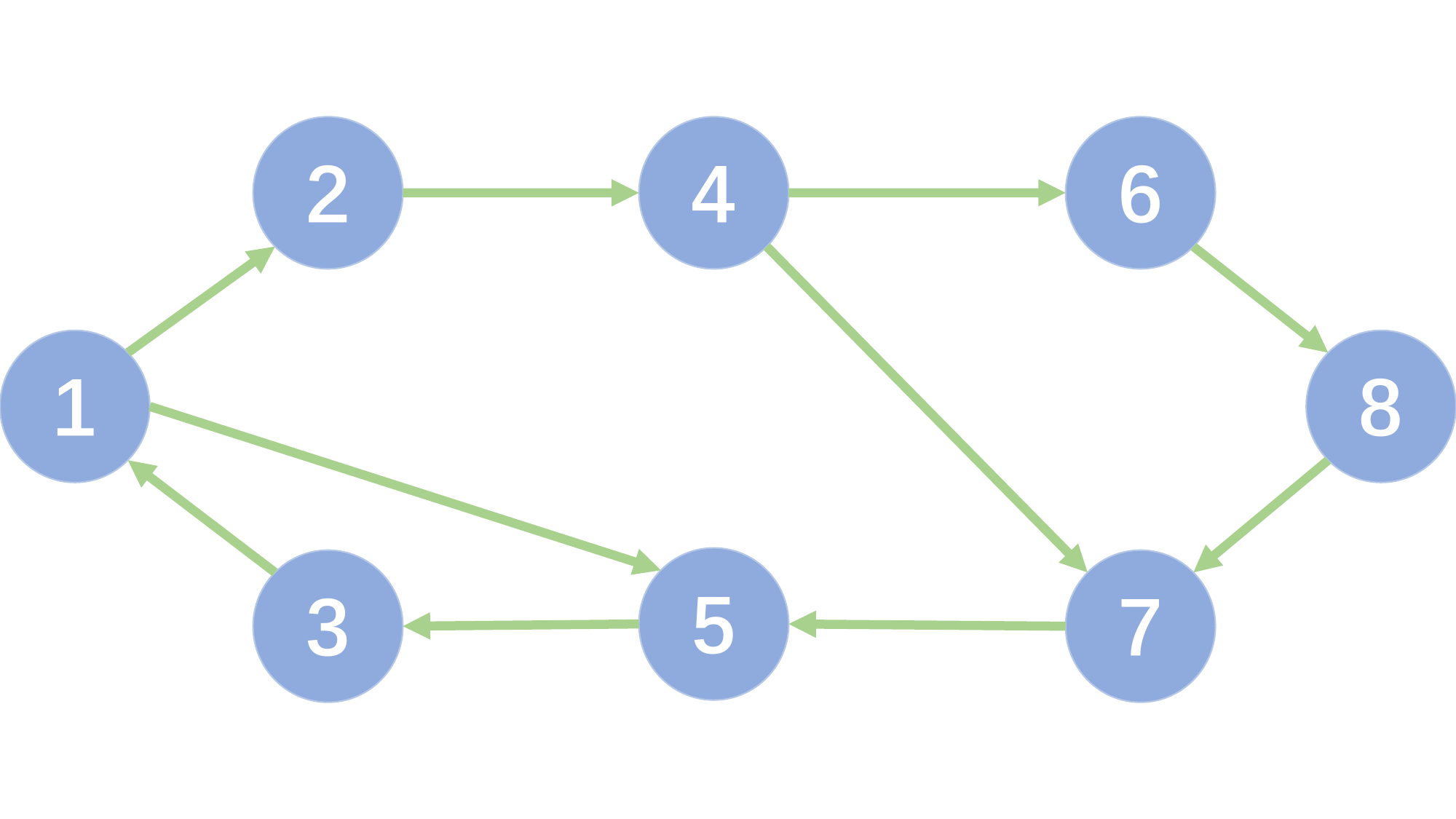}
		\caption{The directed communication network.}
		\label{fig:communication8}
	\end{figure}
	
	Fig. \ref{NESeekingGraph} shows the squared errors from Algorithm \ref{Algorithm 2} with constant and diminishing step-size, as well as the errors from the algorithm in \cite{02Fang2022Directed} under the same gradient computations. \textcolor{blue}{It illustrates that, given the same number of gradient computations, Algorithm \ref{Algorithm 2} with a constant step-size attains higher accuracy compared to the existing algorithm in \cite{02Fang2022Directed}. Moreover, with a diminishing step-size, Algorithm \ref{Algorithm 2} converges to NE more rapidly while achieving similar accuracy, as indicated by the red curve stabilizing faster than the blue one.} Table \ref{NESeekingTable} shows the final error of the algorithms with same communication numbers. The comparative results indicate that Algorithm \ref{Algorithm 2}, even with the same communication numbers, can achieve a smaller error. Moreover, gradient computation number is only half of that of algorithm in \cite{02Fang2022Directed}.
	
	\begin{figure}[H]
		\centering
		\includegraphics[width=0.9\linewidth]{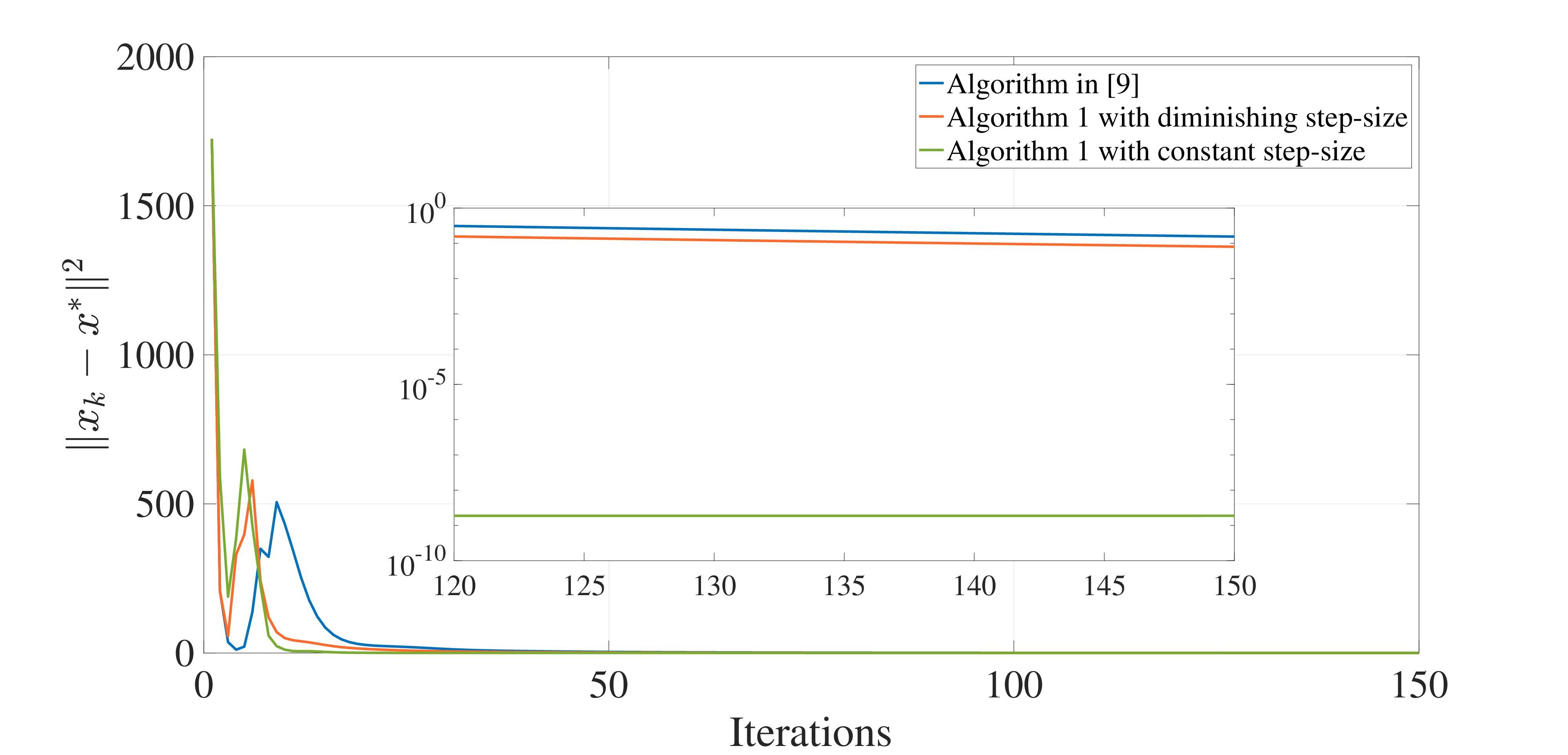}
		\caption{The squared errors of different algorithms with same gradient computation times.}
		\label{NESeekingGraph}
	\end{figure}
	\begin{table}[H]
		\centering
		\Huge
		\caption{The squared errors by different algorithms}
		\label{NESeekingTable}
		\resizebox{\linewidth}{!}{
		\begin{tabular}{|c|c|c|}
			\hline
			& 50 times gradient computation & 150 times communications \\
			\hline
			\makecell{Algorithm 1 with \\$\alpha_k = 5$ and $S_k = 2$} & $1.686487017104607\times 10 ^{-4}$ & $2.234054136428310\times 10 ^{-7}$ \\
			\hline
			\makecell{Algorithm 1 with \\$\alpha_k = \frac{8}{k^{0.65}}$ and $S_k = 2$} & $1.650012506016668$ & $5.83309863858811\times 10 ^{-1}$ \\
			\hline
			algorithm in \cite{02Fang2022Directed} & $3.541863670177937$ &  $1.55742187346232\times 10 ^{-1}$\\
			\hline
		\end{tabular}}
	\end{table}
	
	Figs. \ref{fig:curvefittingfixed} and \ref{fig:curvefittingdiminishing} show the squared errors $\|x_k-x^*\|^2$ with constant/diminishing step-sizes and $S_k = 2$, respectively. The blue curves are obtained by fitting using the MATLAB fitting toolbox. As depicted in Fig. \ref{fig:curvefittingfixed}, the decay of the squared error with constant step-size closely matches the blue curve, which decays at a rate of $O(a^k)$. Similarly, Fig. \ref{fig:curvefittingdiminishing} demonstrates that the decay of the squared error with diminishing step-size closely matches the blue curve, decaying at a rate of $O(\frac{1}{(k+1)^{2r-1}})$.
	\begin{figure}[H]
		\centering
		\includegraphics[width=0.9\linewidth]{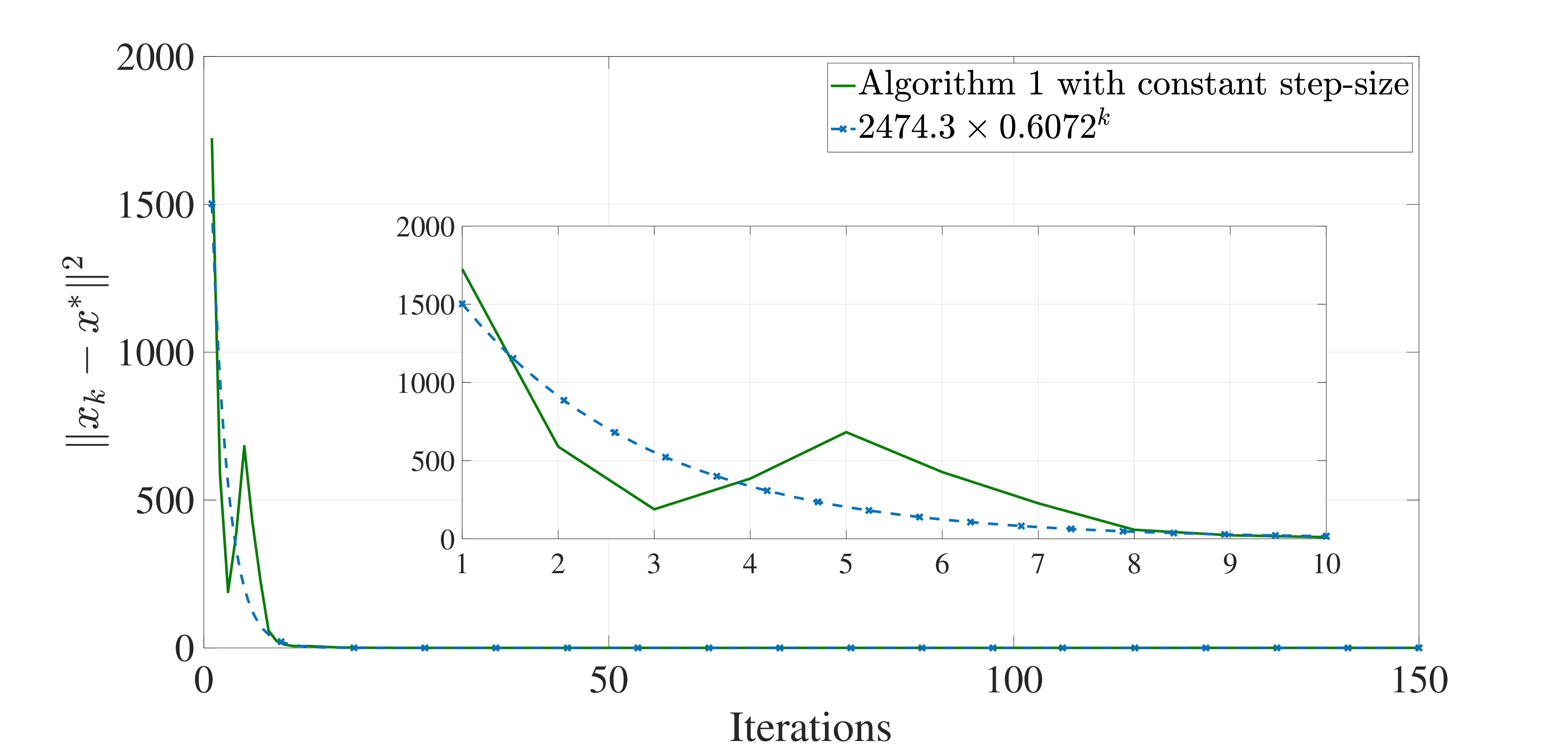}
		\caption{The squared errors by Algorithm \ref{Algorithm 2} with constant step-size.}
		\label{fig:curvefittingfixed}
	\end{figure}
	\begin{figure}[H]
		\centering
		\includegraphics[width=0.9\linewidth]{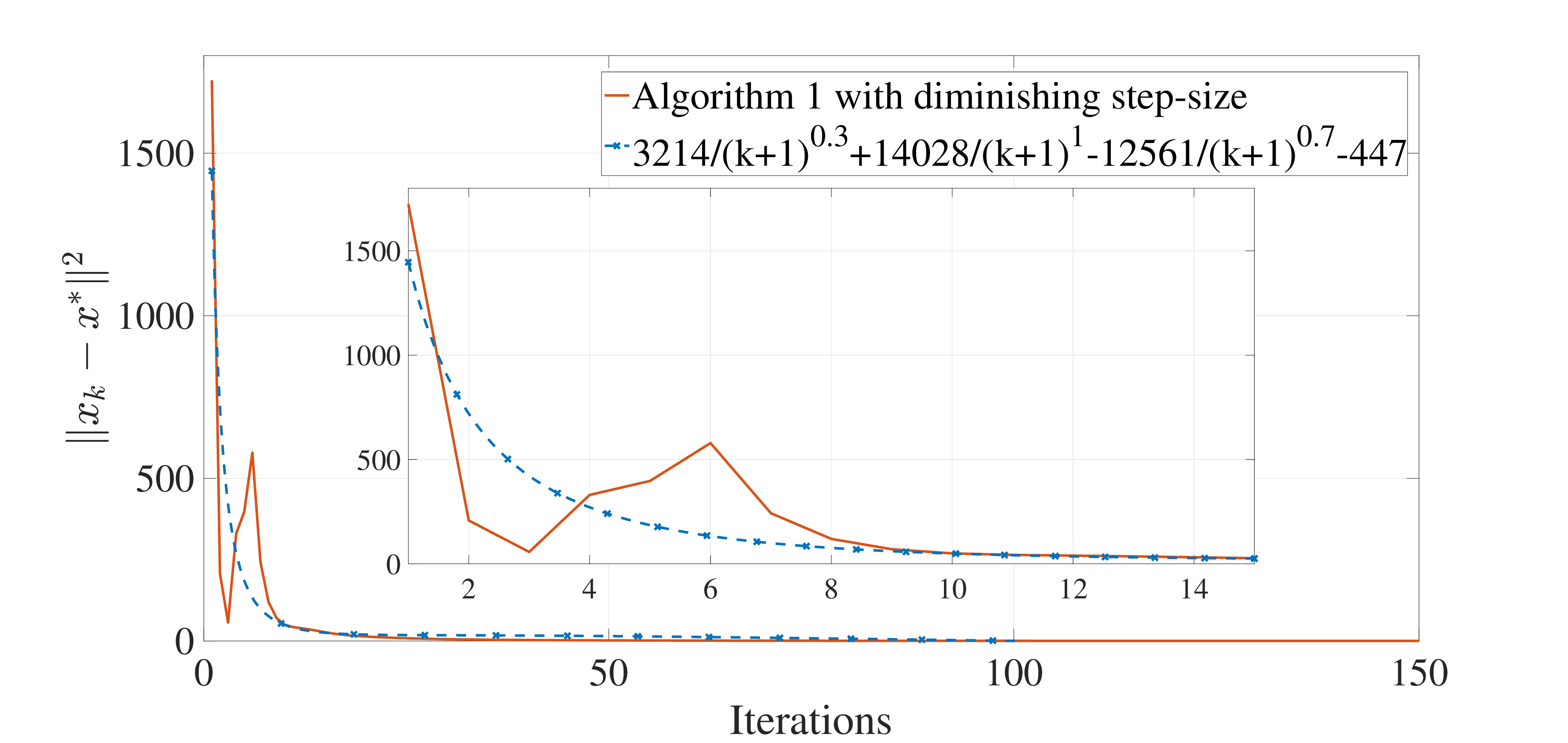}
		\caption{The squared errors by Algorithm \ref{Algorithm 2} with diminishing step-size.}
		\label{fig:curvefittingdiminishing}
	\end{figure}
	
	\subsection{Verification for the resilient fast NE seeking algorithm}
	This subsection verifies the effectiveness of Algorithm \ref{Algorithm 1} by introducing two malicious players into the game. The resulting communication network is depicted in Fig. \ref{fig:communication10}. The trust observations $\tau_{ij,k,l}$ are generated from uniform distributions with the following four types: $\textcircled{\scriptsize{1}}[0.25, 0.85], \textcircled{\scriptsize{2}}[0.4, 1], \textcircled{\scriptsize{3}}[0.15, 0.75], \textcircled{\scriptsize{4}}[0.05, 0.65]$.
	
	Fig. \ref{fig:resineseeking} shows the errors by using Algorithm \ref{Algorithm 1}. Initially, all neighbors are considered to be legitimate, and the subsequent changes in classification are presented in Table \ref{Changes in classification}. For instance, at iteration $2$, player $1$ classifies $3$ as malicious. After iteration $20$, the classification of neighbors by all players remains unchanged. Fig. \ref{fig:resineseekingmean} shows the mean squared errors by running Algorithm \ref{Algorithm 1} 2000 times. It demonstrates that Algorithm \ref{Algorithm 1} can converge to the NE with a small error, thereby validating Theorem \ref{SectionTheorem3}.
	
	\begin{figure}[H]
		\centering
		\includegraphics[width=0.6\linewidth]{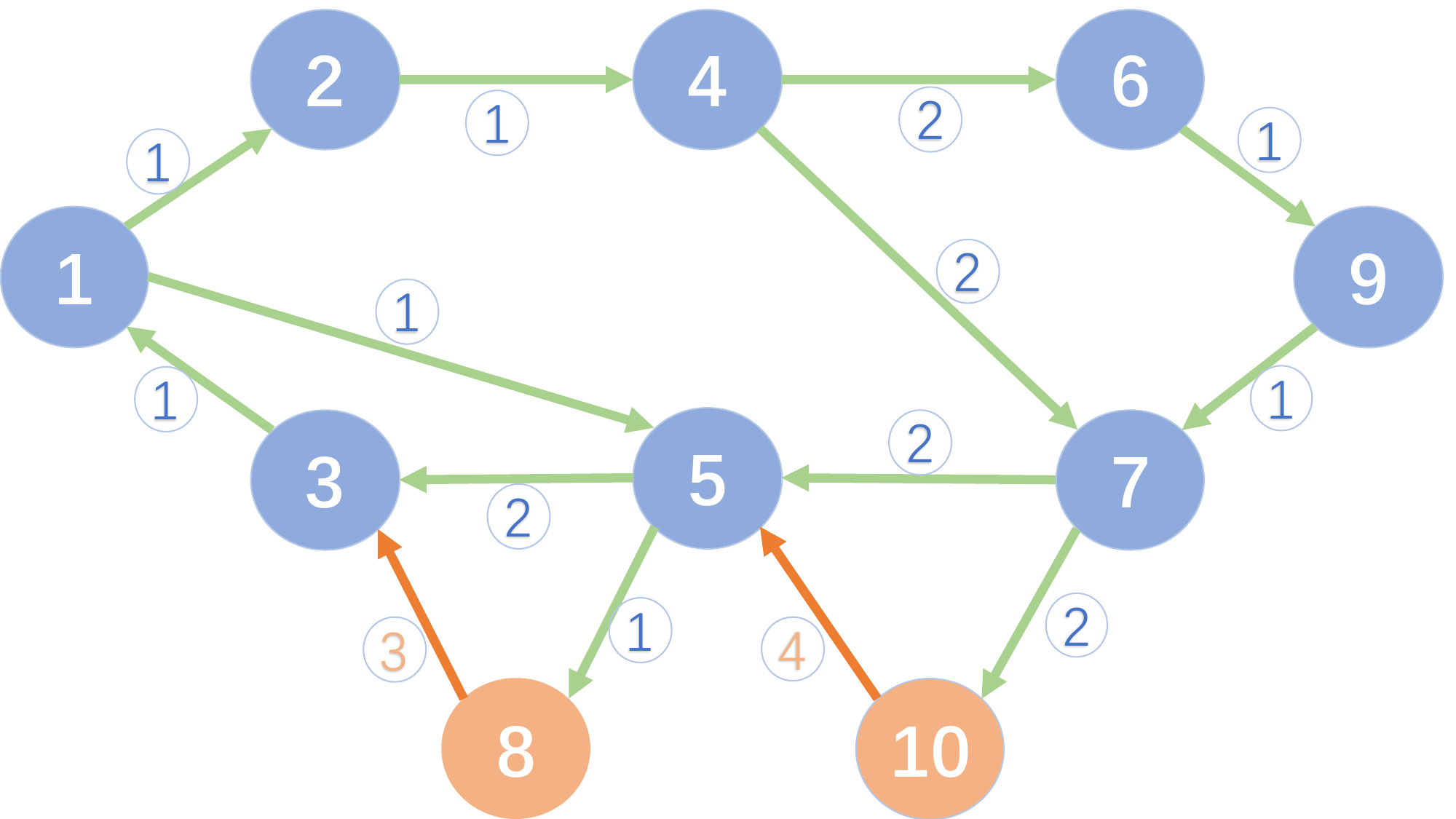}
		\caption{The communication network with malicious players 8 and 10.}
		\label{fig:communication10}
	\end{figure}
	\begin{figure}[H]
		\centering
		\includegraphics[width=1\linewidth]{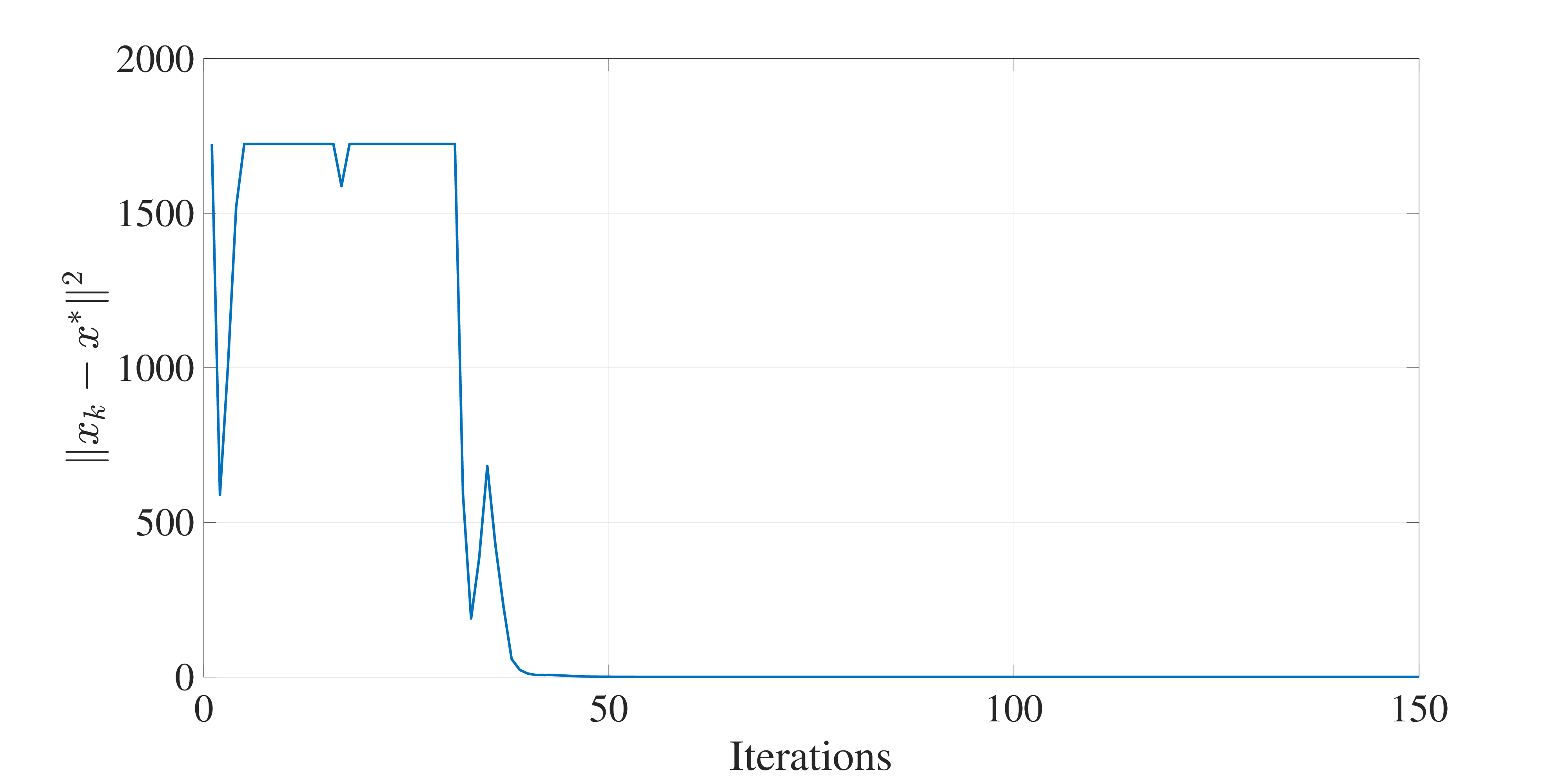}
		\caption{The squared errors by Algorithm \ref{Algorithm 1} with constant step-size.}
		\label{fig:resineseeking}
	\end{figure}
	
\begin{table}[H]
	\small
	\caption{Changes in classification.}
	\label{Changes in classification}
	\begin{minipage}[l]{0.45\linewidth}
		\centering
		\begin{tabular}{|c|c|c|c|}
			\hline
			k & i & j & \makecell{classifi-\\cation} \\
			\hline
			2 & 1 & 3 & $\mathcal{M}$ \\
			\hline
			2 & 5 & 10 & $\mathcal{M}$ \\
			\hline
			2 & 8 & 5 & $\mathcal{M}$ \\
			\hline
			3 & 1 & 3 & $\mathcal{L}$ \\
			\hline
			3 & 3 & 8 & $\mathcal{M}$ \\
			\hline
			3 & 4 & 2 & $\mathcal{M}$ \\
			\hline
		\end{tabular}
	\end{minipage}
	\begin{minipage}[r]{0.45\linewidth}
		\centering
		\begin{tabular}{|c|c|c|c|}
			\hline
			k & i & j & \makecell{classifi-\\cation} \\
			\hline
			5 & 4 & 2 & $\mathcal{L}$ \\
			\hline
			7 & 8 & 5 & $\mathcal{L}$ \\
			\hline
			13 & 3 & 8 & $\mathcal{L}$ \\
			\hline
			14 & 3 & 8 & $\mathcal{M}$ \\
			\hline
			16 & 3 & 8 & $\mathcal{L}$ \\
			\hline
			20 & 3 & 8 & $\mathcal{M}$ \\
			\hline
		\end{tabular}
	\end{minipage}
\end{table}

\begin{figure}[H]
	\centering
	\includegraphics[width=1\linewidth]{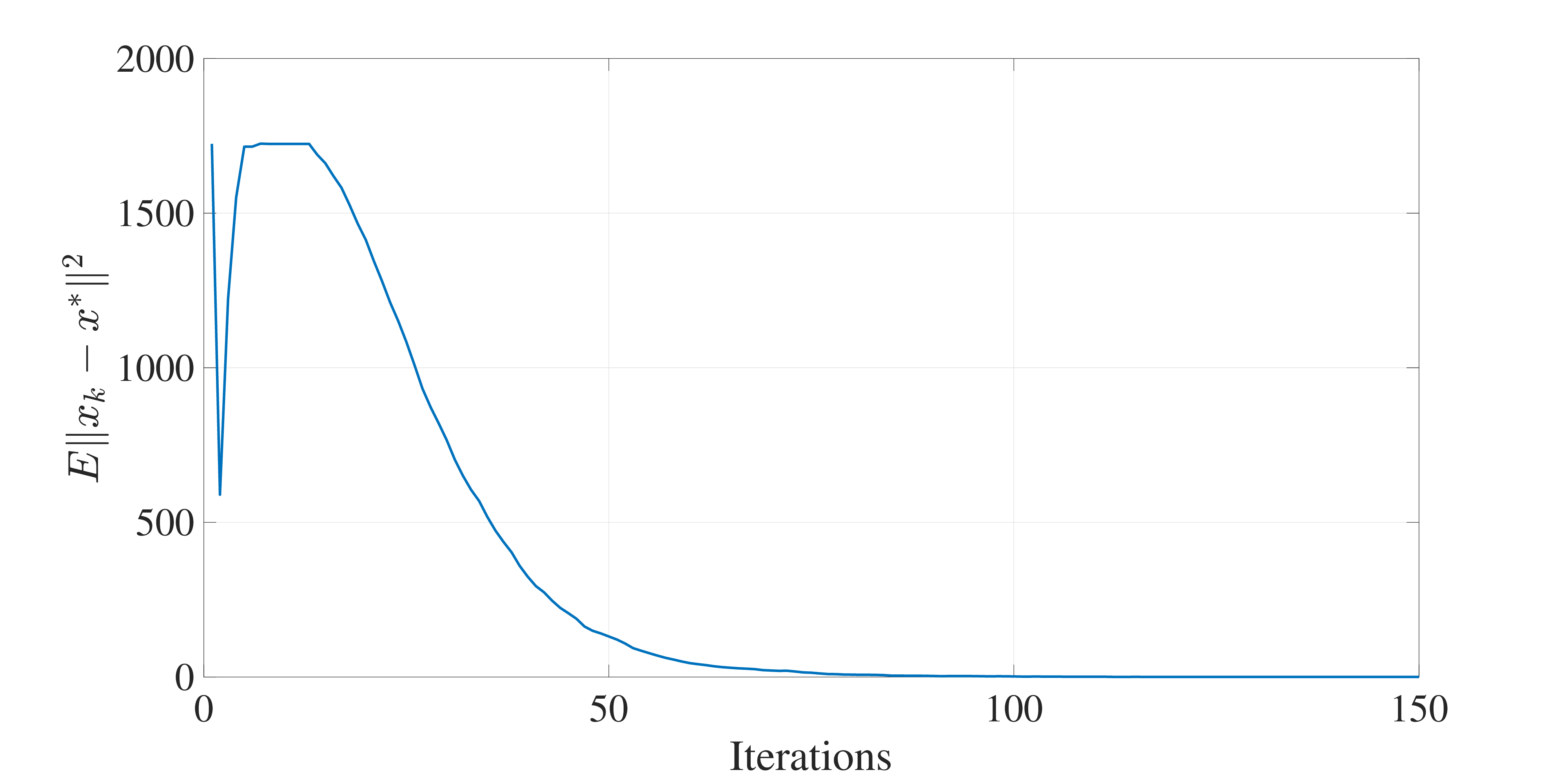}
	\caption{The mean squared errors by Algorithm \ref{Algorithm 1} with constant step-size.}
	\label{fig:resineseekingmean}
\end{figure}

\textcolor{blue}{While the case study is conducted under ideal communication conditions to clearly illustrate the performance in terms of mean sqaure errors and convergence speed, real-world aggregative games may involve communication delays, packet loss, or asynchronous updates. These factors can affect neighbor identification, as well as the convergence rate and accuracy of the NE seeking process.}
	
\section{conclusion}
	In this paper, we have proposed a resilient fast NE seeking algorithm for aggregative games with unbalanced directed network. We have introduced a trustworthiness probabilistic framework that describes the behavior of players and extended it to account for heterogeneity among players. By employing a constant step-size and a compressible unbalanced network information matrix, we have derived a linearly convergent NE seeking algorithm for the game without  malicious players. By integrating a multi-round communication mechanism and diminishing step-size, we have achieved an algorithm with rigorous convergence. Furthermore, by integrating the multi-round communication mechanism and a trustworthiness broadcast mechanism, we have embedded the trustworthiness probabilistic framework into the fast NE seeking algorithm for the game with malicious players, resulting in a resilient fast NE seeking algorithm, and analyzed its convergence. \textcolor{blue}{The proposed algorithm may be further improved by leveraging the momentum term presented in \cite{chen2024achieving} to improve the accuracy or by employing the techniques in \cite{nguyen2025distributed} to adapt to time-varying networks. Future work will also focus on extending the probabilistic trustworthiness framework to settings with unbounded trust observations and adaptive malicious players, as well as incorporating real-world constraints into the proposed method.} 
	
\section*{Acknowledgment}
{The authors would like to thank Prof. Angelia Nedi\'c from Arizona State University for her help and relevant discussions.}

\section*{References}
\bibliographystyle{IEEEtran}
\bibliography{reference}

\begin{IEEEbiography}[{\includegraphics[width=1in,height=1.25in,clip,keepaspectratio]{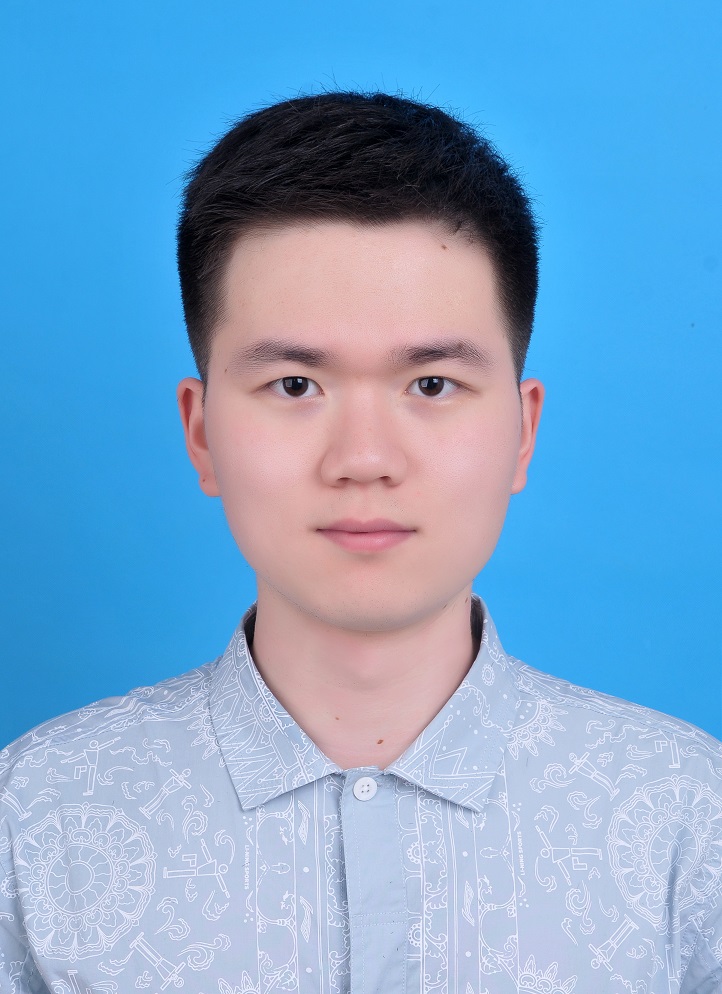}}]{Kai-Yuan Guo}received the B.S. degree from Huazhong University of Science and Technology, Wuhan, China, in 2022. He is currently working toward the Ph.D degree in control theory and control engineering with the Huazhong University of Science and Technology, Wuhan, China. His research interests include game theory and distributed optimization in the smart grid.
\end{IEEEbiography}

\begin{IEEEbiography}[{\includegraphics[width=1in,height=1.25in,clip,keepaspectratio]{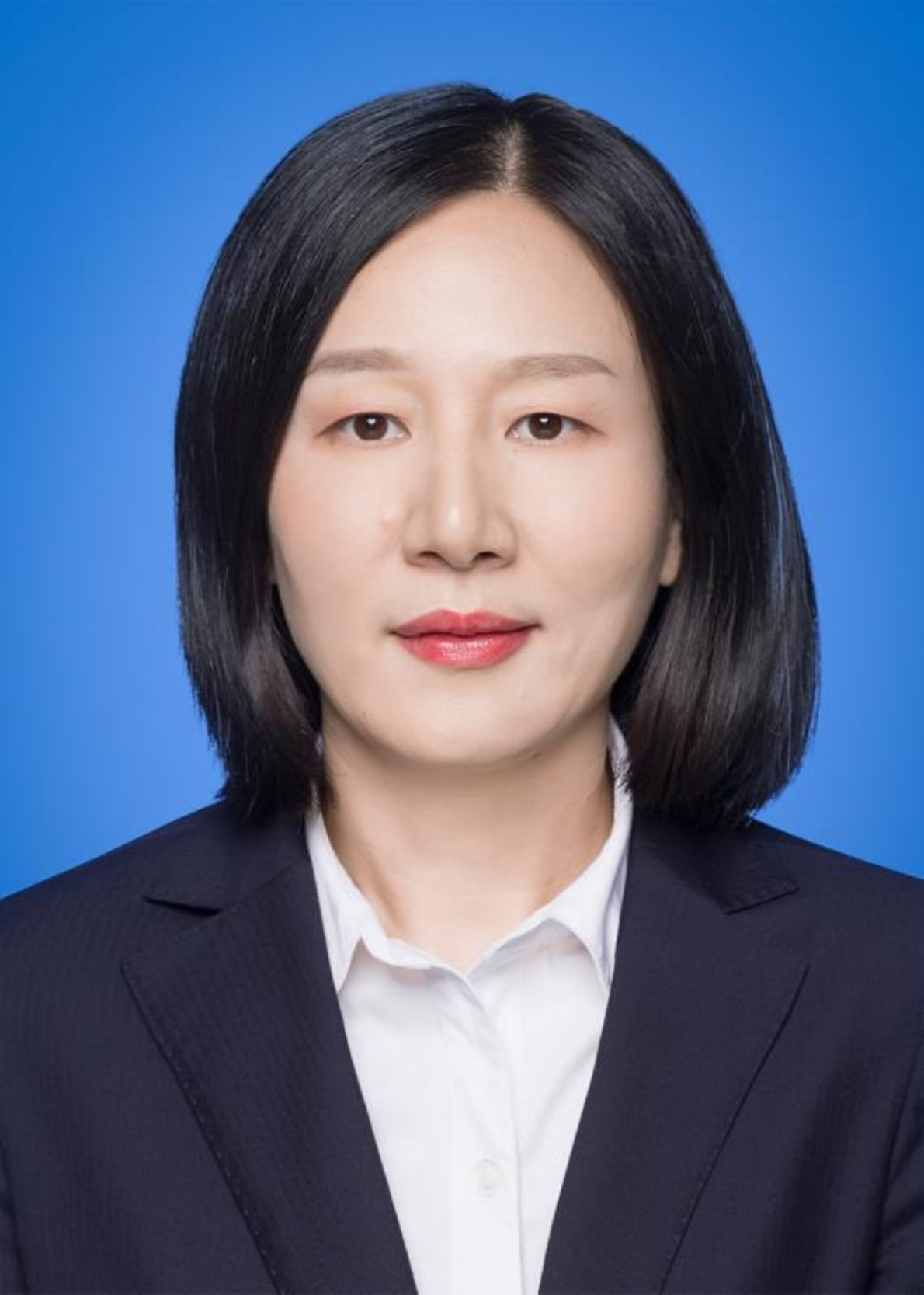}}]{Yan-Wu Wang}
	(M'10--SM'13) received the B.S. degree in automatic control, the M.S. degree and the Ph.D. degree in control theory and control engineering from Huazhong University of Science and Technology (HUST), Wuhan, China, in 1997, 2000, and 2003, respectively. She has been a Professor with the School of Artificial Intelligence and Automation, HUST, since 2009. Currently, she is also with the
	Key Laboratory of Image Processing and Intelligent Control, Ministry of Education, China. Her research interests include hybrid systems, cooperative control, and multi-agent systems with applications in the smart grid.
	
	Dr. Wang was a recipient of several awards, including the first prize of Hubei Province Natural Science in 2014, the first prize of the Ministry of Education of China in 2005, and the Excellent PhD Dissertation of Hubei Province in 2004, China. In 2008, she was awarded the title of ``New Century Excellent Talents" by the Chinese Ministry of Education.
\end{IEEEbiography}

\begin{IEEEbiography} [{\includegraphics[width=1in,height=1.25in,clip,keepaspectratio]{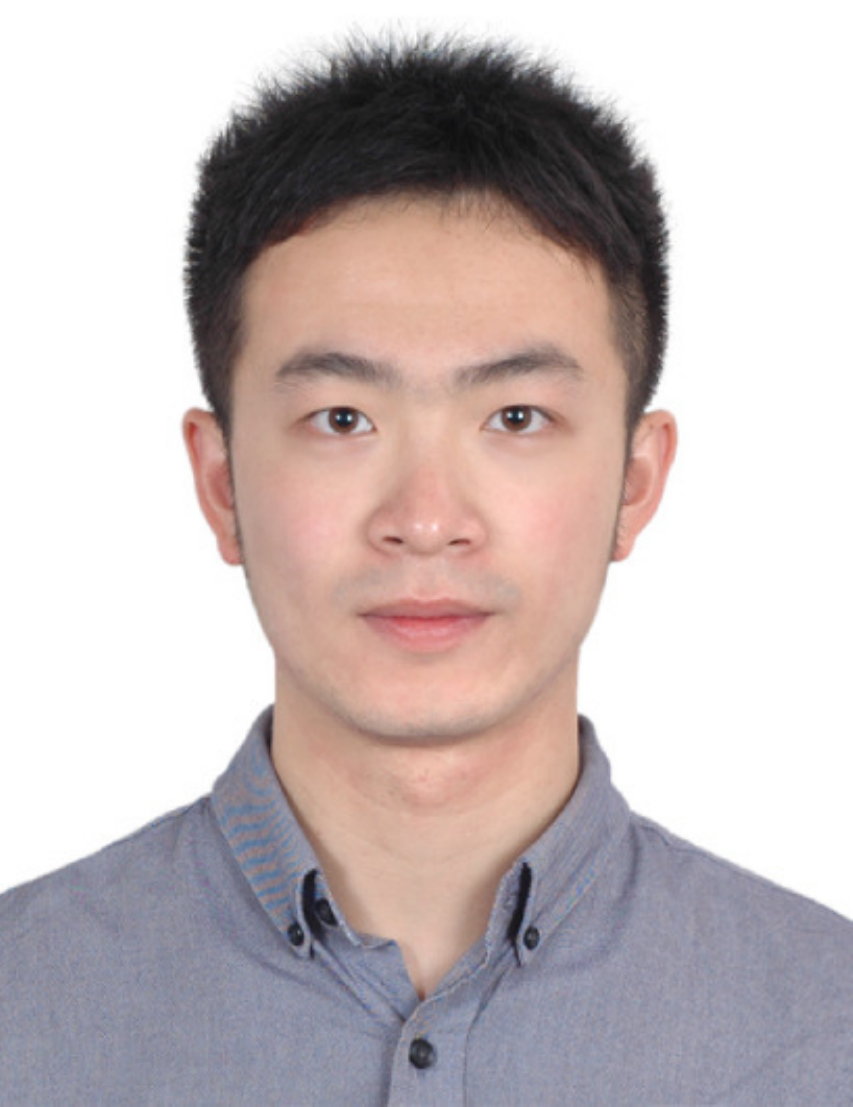}}] {Xiao-Kang Liu}(M'20) received the B.S. degree in automatic control and the Ph.D. degree in control science and engineering from Huazhong
	University of Science and Technology (HUST), Wuhan, China, in 2014 and 2019, respectively. From Jul. 2017 to Aug. 2018, he was a visiting scholar with the Department of Electrical, Computer, and Biomedical Engineering, University of Rhode Island, RI, USA. From Oct. 2019 to Feb. 2021, he was a postdoctoral research fellow with the School of Electrical \& Electronic Engineering, Nanyang Technological University (NTU), Singapore. He is currently working as an Associate Professor with the School of Artificial Intelligence and Automation, HUST. His research interests include hybrid control, distributed control and optimization, DC microgrids.
\end{IEEEbiography}

\begin{IEEEbiography}
	[{\includegraphics[width=1in,height=1.25in,clip,keepaspectratio]{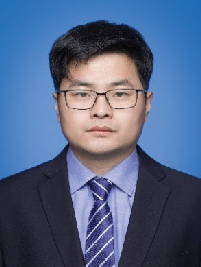}}]{Zhi-Wei Liu} (M'14--SM'22) received the B.S. degree in Information Management and Information System from Southwest Jiaotong University, Chengdu, China, in 2004, and the Ph.D. degree in Control Science and Engineering from the Huazhong University of Science and Technology, Wuhan, China, in 2011.
	
	He is currently a Professor with the School of Artificial Intelligence and Automation at the Huazhong University of Science and Technology, Wuhan, China. His current research interests include cooperative control and optimization of distributed network systems.
\end{IEEEbiography}

\end{document}